\documentclass[12pt,onecolumn]{IEEEtran}

\usepackage{amsmath,amsfonts,amssymb,amsthm}
\usepackage{mathtools,bm}
\usepackage{algorithmic}
\usepackage{algorithm}
\usepackage{array}
\usepackage[caption=false,font=normalsize,labelfont=sf,textfont=sf]{subfig}
\usepackage{textcomp}
\usepackage{stfloats}
\usepackage{url}
\usepackage{verbatim}
\usepackage{graphicx}
\usepackage{booktabs}
\usepackage{cite}
\usepackage{color}

\theoremstyle{plain}
\newtheorem{theorem}{Theorem}[section]
\newtheorem{lemma}[theorem]{Lemma}

\newtheorem{corollary}[theorem]{Corollary}

\newtheorem*{main}{Main Theorem}
\newtheorem*{lemma*}{Lemma}
\newtheorem*{corollary*}{Corollary}
\theoremstyle{definition}
\newtheorem{definition}[theorem]{Definition}

\newcommand{\repname}{}
\newtheorem*{repinner}{\repname}

\newenvironment{rep}[2]
{%
	\renewcommand{\repname}{#1~\ref{#2}}%
	\begin{repinner}%
	}
	{%
	\end{repinner}%
}

\theoremstyle{remark}
\newtheorem{remark}[theorem]{Remark}

\allowdisplaybreaks

\newcommand{\cC}{{\mathcal{C}}}

\newcommand{\bc}{{\mathbf{c}}}
\newcommand{\bg}{{\mathbf{g}}}

\newcommand{\bv}{{\mathbf{v}}}

\newcommand{\bx}{{\mathbf{x}}}
\newcommand{\by}{{\mathbf{y}}}
\newcommand{\bz}{{\mathbf{z}}}
\newcommand{\bh}{{\mathbf{h}}}
\newcommand{\be}{{\mathbf{e}}}

\newcommand{\tH}{{\mathbf{H}}}
\newcommand{\tI}{{\mathbf{I}}}
\newcommand{\tG}{{\mathbf{G}}}

\newcommand{\tP}{{\mathbf{P}}}

\newcommand{\tD}{{\mathbf{D}}}

\newcommand{\tM}{{\mathbf{M}}}
\newcommand{\tN}{{\mathbf{N}}}

\newcommand{\bzero}{{\mathbf{0}}}

\newcommand{\PAut}{{\operatorname{PAut}}}
\newcommand{\MAut}{{\operatorname{MAut}}}
\newcommand{\Gr}{{\operatorname{Gr}}}

\newcommand{\ord}{{\operatorname{ord}}}
\newcommand{\diag}{{\operatorname{diag}}}
\newcommand{\im}{{\operatorname{Im}}}
\newcommand{\Ker}{{\operatorname{Ker}}}
\newcommand{\Char}{{\operatorname{char}}}
\newcommand{\gf}{{\mathbb{F}}}
\newcommand{\wt}{{\operatorname{wt}}}
\newcommand{\Aut}{{\operatorname{Aut}}}

\def\qi#1 {\fbox {\footnote {\ }}\ \footnotetext { From Qi: {\color{red}#1}}}

\begin{document}
	
	\title{The automorphism groups of random linear codes}
	
	\author{Xiaoru Li, Qi Wang, and Yue Zhou \thanks{X. Li and Q. Wang were in part supported by the National Key Research and Development Program of China (No. 2025YFA1017700) and the National Natural Science Foundation of China (No. 12371522). Y. Zhou was supported by the National Natural Science Foundation of China (No.~12371337). X. Li is with the Department of Computer Science and Engineering, Southern University of Science and Technology, Shenzhen 518055, China (e-mail:
	lixr2024@mail.sustech.edu.cn), Q. Wang is with the Department of Computer Science and Engineering, Southern University of Science and Technology, Shenzhen 518055, China, and also with National Center for Applied Mathematics Shenzhen, Southern University of Science and Technology, Shenzhen 518055, China (e-mail: wangqi@sustech.edu.cn), Y. Zhou is with the College of Science, National University of Defense Technology, Changsha 410073, China (e-mail: yue.zhou.ovgu@gmail.com). (\emph{Corresponding author: Qi Wang})}}

	\maketitle
	
	\begin{abstract}
		The matching codewords framework is a key tool in recent algorithms for solving
		the Linear Code Equivalence (LCE) problem and in security analyses of
		LCE-based cryptographic schemes such as LESS. These analyses often rely
		on the assumption that a random \(q\)-ary linear code has no monomial
		automorphisms other than scalar multiples of the identity. For binary
		codes, Lefmann, Phelps, and R\"odl established the corresponding
		rigidity phenomenon in the relevant logarithmic dimension range. For
		general \(q\), Hou established an averaged result over all dimensions,
		whereas the recent prescribed-dimension result of Di Giusto and
		Ravagnani applies only in a restricted regime near \(n/2\).
		
		For every fixed prime power $q$ and every fixed real number $\varepsilon>0$, we prove that a uniformly random $k$-dimensional code $\mathcal{C}\subseteq\mathbb{F}_q^n$ has a trivial monomial automorphism group with probability tending to $1$ as $n\to\infty$, provided that $m:=\min\{k,n-k\}\geq(2+\varepsilon)\log_q n$.
		Furthermore, when $m \le 2 \log_q n + C$, where $C$ is a constant independent of $n$, we also show that the probability that the automorphism group of $\cC$ is nontrivial is at least $\frac{1}{2} - \varepsilon$ for large enough $n$. 
	\end{abstract}
	
	
	\section{Introduction}\label{sec1}
	
	Let $\gf_q$ denote the finite field with $q$ elements, where $q$ is a prime power. Let $n$ be a positive integer, and $[n] = \{1, 2, \ldots, n\}$. An $[n,k]_q$ linear code $\cC$ is defined as a $k$-dimensional $\gf_q$-linear subspace of $\gf_q^n$. The vectors $\bc=(c_1, c_2, \dots, c_n) \in \cC$ are called {\em codewords} of $\cC$. The equivalence relation of linear codes is usually defined by linear isometries that preserve the distance. When linear codes are endowed with the Hamming metric, there exist three main notions of equivalence, each corresponding to a group action: permutation equivalence, monomial equivalence, and semilinear equivalence. Note that for binary linear codes, these three notions coincide. 
	
	The symmetric group $S_n$ is the group of all permutations on $[n]$, and acts on the vectors of $\gf_q^n$ by permuting the coordinate positions. All monomial isometries from $\gf_q^n$ to $\gf_q^n$ form a group $M_{n,q}=(\mathbb F_q^*)^n \rtimes S_n$, which is called the {\em monomial group}. In particular, for any $\sigma=(\bv; \pi) \in M_{n,q}$ and for any codeword $\bx=(x_1, x_2, \dots, x_n) \in \gf_q^n$, 
	\[\sigma(\bx)=(v_1x_{\pi^{-1}(1)}, v_2x_{\pi^{-1}(2)}, \dots, v_nx_{\pi^{-1}(n)}),\] 
	where $\bv=(v_1, v_2, \dots, v_n)$.
	For two linear codes $\cC_1$ and $\cC_2$, if there exists a linear isometry $\sigma \in M_{n,q}$ (resp., $\sigma \in S_n$) such that $\cC_2=\sigma(\cC_1)$, then $\cC_1$ and $\cC_2$ are said to be {\em monomially equivalent} (resp., {\em permutation equivalent}). More generally, two linear codes $\cC_1$ and $\cC_2$ are called {\em semilinearly equivalent}, if there exists a semilinear isometry $(\bv; (\alpha, \pi)) \in (\gf_q^*)^n \rtimes (\Aut(\gf_q) \times S_n)$ that maps $\cC_1$ onto $\cC_2$, i.e., 
	\[\cC_2=(\bv; (\alpha, \pi))(\cC_1)=\{(\bv; (\alpha, \pi))(\bc): \bc=(c_1, c_2, \dots, c_n) \in \cC_1\},\] 
	where 
	\[(\bv; (\alpha, \pi))(\bc)=(v_1\alpha(c_{\pi^{-1}(1)}), \dots, v_n\alpha(c_{\pi^{-1}(n)})).\] 
	In the contexts of cryptography, semilinear equivalence are usually not considered since the automorphism group $\Aut(\gf_q)$ is trivial for $q$ prime. For a linear code $\cC$, the permutation automorphism group is defined as $\PAut(\cC)=\{\pi \in S_n: \pi(\cC)=\cC\}$, and the monomial automorphism group is defined as $\MAut(\cC)=\{\sigma \in M_{n,q}: \sigma(\cC)=\cC\}$. For any $\sigma_{\alpha}=\alpha\tI_n \in M_{n,q}$ with $\alpha \in \gf_q^*$, it is clear that any linear code $\cC$ is $\sigma_{\alpha}$-stable, and  $\MAut_0:=\gf_q^*\tI_n \subseteq \MAut(\cC)$ for any linear code $\cC$.
	We say that a linear code $\cC \subseteq \gf_q^n$ has a {\em trivial} monomial automorphism group if $\MAut(\cC)=\MAut_0$. For simplicity, if not stated otherwise, hereafter the automorphism group is referred to as the monomial automorphism group.
	
	The interaction between codes and group actions is a central theme in coding theory because an automorphism group can reveal structural properties of a code~\cite[Ch. 17]{Huffman98}.  Automorphism groups also play important roles in code classification~\cite{Leon, Feulner}, as well as in the design of efficient decoding algorithms~\cite[Ch. 17]{Huffman98,GEECB21a}. Leon~\cite{Leon} gave an algorithm for determining the automorphism groups of codes and deciding whether two codes are equivalent. Another algorithm for computing the automorphism groups of codes was given by Feulner~\cite{Feulner}, and was able to calculate a canonical representative of a semilinear equivalence class of codes. 
	In addition, the automorphism groups of special classes of codes, such as cyclic codes~\cite{BK10, GG13, GG17, FHLX26}, generalized Reed-Muller codes~\cite{BC93}, and polar codes~\cite{GEECB21, MG24}, have been extensively studied. 
	
	The Linear Code Equivalence (LCE) problem has received considerable attention in post-quantum cryptography in recent years. The decision version of the LCE problem asks, given two linear codes $\cC_1$ and $\cC_2$, whether there exists $\sigma \in M_{n,q}$ such that $\cC_2=\sigma(\cC_1)$. Its search version asks to explicitly find such an isometry $\sigma$. The search version of the LCE problem is closely related to public key cryptosystems, such as the McEliece framework~\cite{McE}, and the LESS scheme~\cite{LESS}. Algorithms for solving the LCE problem are important for analyzing the security of the LESS scheme and especially the choice of its parameters. Currently, the matching codewords framework has become a key tool for analyzing both the concrete and asymptotic complexities of the LCE problem~\cite{Leon, Ward20, BBPS23, BEFN26, BE26, BBOSS26}. The core idea of this framework is to identify matching codewords between two equivalent codes using the Information Set Decoding (ISD) algorithm~\cite{Prange62}, thereby gradually recovering the underlying isometry through codeword matching. This framework often relies on the assumption that $q$-ary random codes have trivial automorphism groups~\cite{Ward20, BBPS23}, i.e., the LCE instances have a unique solution (up to scalars)~\cite{BEFN26, BE26, BBOSS26}.
	
	Although in the literature, it is widely assumed that random $q$-ary codes have trivial automorphism groups, this assumption has not been formally proved~\cite{BBPS23, HMSW26}. In the binary case, Oral and Phelps showed in~\cite{HP92} that the automorphism group of a binary self-dual code is trivial with high probability. Subsequently, Lefmann, Phelps, and R\"odl~\cite{LPR93} proved that the automorphism group of a binary random code is trivial with high probability. For $q$-ary random linear codes, Hou~\cite{Hou05} studied the asymptotic number of non-equivalent codes by employing Burnside's lemma, and showed that the proportion of all subspaces of $\mathbb{F}_q^n$ with a nontrivial automorphism group tends to zero as $n$ goes to infinity. However, this is an averaged result over all dimensions, and it does not directly imply that a uniformly random linear code of $k$-dimension has a trivial automorphism group with high probability. Very recently, Di Giusto and Ravagnani~\cite{GR26} investigated the asymptotic number of equivalence classes of linear codes with a given dimension. Their results imply that a random $k$-dimensional subspace of $\gf_q^n$ has a trivial automorphism group with high probability. However, this conclusion only holds when the dimension $k$ is within a restricted regime near $n/2$, i.e., $|k-\frac{n}{2}|=O(\sqrt{n})$. In general, the asymptotic behavior of the automorphism groups of random $q$-ary linear codes remains an open problem.
	
	In this paper, we address this problem by estimating the asymptotic probability of a random $k$-dimensional subspace of $\mathbb{F}_q^n$ having a nontrivial automorphism group. Our main result is as follows.
	\begin{main}
		Fix a prime power $q$ and a real number $\varepsilon>0$. Let $\cC \subseteq \gf_q^n$ be a uniformly random $k$-dimensional linear code, and let $m :=\min\{k, n-k\}$. If $m\ge (2+\varepsilon)\log_q n$,
		then $\cC$ has a trivial automorphism group with probability tending to one as $n$ goes to infinity.
	\end{main}
	
	As the counterpart of the case of trivial automorphism groups, when $m$ is comparably small, zero columns or proportional columns are more likely to appear in its generator matrix or parity-check matrix. This may lead to the existence of nontrivial automorphism groups. In particular, we further prove that the probability that a random $k$-dimensional linear code $\mathcal{C}$ has a nontrivial automorphism group is at least $\frac{1}{2} - \varepsilon$ whenever $m$ is upper bounded by $2 \log_q n + C$ with $C$ a constant independent of $n$ (see Theorem~\ref{thm-nontrivial}).
	
	The remainder of this paper is organized as follows. In Section \ref{sec2}, we introduce some notations and some important preliminary results that will be used throughout this paper. In Section \ref{sec3}, we deal with the case of nontrivial automorphisms for random $k$-dimensional linear codes. In Section \ref{sec4}, we prove the main theorem by counting all possible cases of invariant subspaces under monomial isometries of prime order with respect to the quotient group. Finally, in Section \ref{sec5}, we conclude the paper with some open problems.

	\section{Preliminaries}\label{sec2}
	
	\subsection{Notations}
	Let $q$ be a prime power, and let $\gf_q$ be the finite field with $q$ elements. We denote by $\gf_q^*$ its multiplicative group, i.e., $\gf_q^*=\gf_q \setminus \{0\}$. The identity matrix of order $n$ is denoted by $\tI_n$. Matrices and vectors are written in bold uppercase and lowercase letters, respectively. Vectors are assumed to be row vectors. For $n \in \mathbb{N}$, let $[n]:=\{1,2,\dots,n\}$. Sets are denoted by uppercase letters and for a set $S$, we denote by $|S|$ its cardinality. Denote by $S_n$ the symmetric group on $n$ elements, i.e., the group of permutations of $[n]$, and $M_{n,q} = (\gf_q^{*})^n \rtimes S_n$ the monomial group, i.e., the group of all monomial isometries from $\gf_q^n$ to $\gf_q^n$. For all asymptotic statements, assume that $q$ is fixed, $n\to\infty$, and $k=k(n)$ may vary with $n$.
	
	\subsection{Linear codes}
	
	In this subsection, we introduce some background in coding theory, and for more details, we refer to~\cite{fundamentalsofECcodes}.
	
	\begin{definition}
		Let $1 \leq k \leq n$ be an integer. An $[n, k]_q$ linear code $\cC$ is a $k$-dimensional linear subspace of $\gf_q^n$. The parameter $n$ is called the {\em length} of $\cC$, $k$ is called the {\em dimension} of $\cC$, and the vectors in the code are called {\em codewords}. A matrix $\tG \in \gf_q^{k \times n}$ is called a {\em generator matrix} of $\cC$ if $\cC=\{\bx\tG: \bx \in \gf_q^k\}$.
	\end{definition}
	
	By convention, we usually write $\langle\tG\rangle$ to denote the linear code $\cC$ generated by the matrix $\tG$. 
	
	\begin{definition}
		Let $1 \leq k \leq n$ be an integer and let $\cC \subseteq \gf_q^n$ be an $[n, k]_q$ linear code. The dual code $\cC^{\perp}$ is an $[n, n-k]_q$ linear code defined as 
		\begin{eqnarray*}
			\cC^{\perp}=\{(v_1,v_2,\ldots, v_n) \in \gf_q^n: \sum_{i=1}^{n}v_ic_i=0, \ \forall \ (c_1,c_2,\ldots, c_n) \in \cC\}.
		\end{eqnarray*}
		A matrix $\tH \in \gf_q^{(n-k) \times n}$ is called a {\em parity-check matrix} of $\cC$ if $\cC^{\perp}=\langle \tH \rangle$.
	\end{definition}
	
	In order to measure how far apart two vectors are, we usually endow $\gf_q^n$ with the Hamming metric.
	
	\begin{definition}
		Let $n$ be a positive integer. The {\em Hamming weight} of a vector $\bx=(x_1, \dots, x_n) \in \gf_q^n$ is defined by the number of its nonzero coordinates, i.e., $\wt_H(\bx)=|\{i \in [n]: x_i \neq 0\}|$. The {\em Hamming distance} between $\bx,\by \in \gf_q^n$ is given by the Hamming weight of their difference vector, i.e., $d_H(\bx, \by)=\wt_H(\bx-\by)$.
	\end{definition}
	
	As the indicator of the error-correcting capability of a linear code  $\cC$, its minimum distance $d_H$ is defined as:
	\begin{eqnarray*}
		d_H(\cC) &:=& \min\{d_H(\bc, \bc'): \bc, \bc' \in \cC \textrm{ and } \bc\neq \bc'\}\\
		& = &\min\{\wt_H(\bc) : \bzero \neq \bc \in \cC\},
	\end{eqnarray*}
	and $\cC$ can correct up to $\lfloor (d_H(\cC)-1)/2\rfloor$ errors.
	
	For $1 \leq k \leq n$, the Grassmannian {$\Gr_q(k, n)$} is the set of all linear codes of dimension $k$ in $\gf_q^n$, and its cardinality is the {\em Gaussian coefficient}, denoted by $\genfrac{[}{]}{0pt}{}{n}{k}_q$. It is well known that
\begin{eqnarray*}
\genfrac{[}{]}{0pt}{}{n}{k}_q = \prod_{j=0}^{k-1} \frac{q^{n-j}-1}{q^{k-j}-1}.
\end{eqnarray*}
 The following lemma provides an estimate for the Gaussian coefficients. 
	
	\begin{lemma}\label{lem-Gau}\cite{I15, NP95}
		For fixed $q$, there exist constants $c_1, c_2$ such that
		\[
		c_1 q^{k(n-k)}
		\le
		\genfrac{[}{]}{0pt}{}{n}{k}_q
		\le
		c_2 q^{k(n-k)}
		\] for all $1 \le k\le n$. 
	\end{lemma}

	\subsection{Code equivalence}
	
	In this subsection, we review the definitions of code equivalence, the automorphism group of codes, and some preliminary results that will be used later.
	
	\begin{definition}
		A linear map $\varphi : \gf_q^n \rightarrow \gf_q^n$ is called an {\em isometry} if it preserves the distance of any two codewords. That is, for all $\bx, \by \in \gf_q^n$, we have $d(\bx, \by)=d(\varphi(\bx), \varphi(\by))$.
	\end{definition}
	
	For $n \geq 3$, all the isometries of the Hamming metric in $\mathbb{F}_q^n$ that map subspaces onto subspaces are exactly the semilinear mappings of the form $(\bv; (\alpha, \pi))$, where $(\bv; \pi) \in M_{n,q}$ is a monomial isometry and $\alpha \in \Aut(\mathbb{F}_q)$~\cite{BBFKKW06}. Throughout this paper, we mainly focus on a subset of the Hamming-metric isometries, namely the monomial group $M_{n,q}$. Any permutation $\pi \in S_n$ can be viewed as an $n \times n$ permutation matrix where $\tP(i, j)=1$ if and only if $j=\pi(i)$. Any monomial 
	map $\sigma=(\bv; \pi) \in M_{n,q}$ can be represented as a monomial matrix $\tM=\tD\tP$, where $\tD=\operatorname{diag}(v_{\pi(1)},\dots, v_{\pi(n)})$ is a diagonal matrix and $\tP$ is the $n \times n$ permutation matrix defined by $\tP(i, j)=1$ if and only if $j=\pi(i)$. By abuse of notation, we write $\tP\in S_n$ if $\tP$ is the permutation matrix corresponding to $\pi\in S_n$, and write $\tM=\tD\tP\in M_{n,q}$ if $\tM$ is the monomial matrix corresponding to $\sigma \in M_{n,q}$.

	\begin{definition}
		Two linear codes $\cC_1$ and $\cC_2$ are {\em monomially equivalent} if there exists a map $\sigma \in M_{n,q}$ such that $\cC_2=\sigma(\cC_1)$. Two linear codes $\cC_1$ and $\cC_2$ are {\em permutation equivalent} if there exists a map $\pi \in S_n$ such that $\cC_2=\pi(\cC_1)$. 
	\end{definition}
	
	The monomial automorphism group of a linear code can be accordingly defined as the set of all the monomial isometries that map it to itself.
	
	\begin{definition}
		Let $\cC$ be an $[n,k]_q$ linear code. The (monomial) automorphism group of $\cC$ is given by the monomial isometries that map $\cC$ to $\cC$, i.e.,
		\begin{eqnarray*}
			\MAut(\cC)=\{\sigma \in (\gf_q^{*})^n \rtimes S_n: \sigma(\cC)=\cC\}.
		\end{eqnarray*}
	\end{definition}
	If for a linear code $\cC$ it holds that $\MAut(\cC) = \MAut_0:=\gf_q^*\tI_n$, then we say that the linear code $\cC$ has a {\em trivial} automorphism group. The following lemma established a link on the automorphism groups of a linear code and its dual code.
	
	
	\begin{lemma}\cite{fundamentalsofECcodes}\label{lem-dualag}
		Let $\cC$ be an $[n,k]_q$ linear code. Then $\MAut(\cC^{\perp})=\{\tD^{-1}\tP : \tD\tP \in \MAut(\cC)\}$.
	\end{lemma}
	
	It then follows from Lemma~\ref{lem-dualag} that $\MAut(\cC) \cong \MAut(\cC^{\perp})$, and it is sufficient to study the automorphism group of linear codes of dimension $k \leq n/2$. The following fundamental theorem in group theory will be used in the proof of the main theorem.
	\begin{lemma}\cite{Algebra}\label{lem-cauchy}
		If $G$ is a finite group whose order is divisible by a prime $p$, then $G$ contains an element of order $p$.
	\end{lemma}
	
	The following lemma, well known as the {\em primary decomposition theorem}~\cite{basic-algebra}, will be used to determine the number of $\tM$-invariant subspaces, where $\tM \in M_{n,q}$ is a monomial isometry.
	
	\begin{lemma}\cite{basic-algebra}\label{lem-PD}
		Let \(\mathcal{L} \colon V \to V\) be linear on a finite-dimensional vector space $V$ over
		\(\mathbb{K}\), and let
		\[
		m(X)=f_1(X)^{l_1}\cdots f_s(X)^{l_s}
		\]
		be the unique factorization of the minimal polynomial \(m(X)\) of \(\mathcal{L}\)
		into the product of powers of distinct monic prime polynomials \(f_j(X)\).
		Define
		$U_j=\Ker\bigl(f_j(\mathcal{L})^{l_j}\bigr)$ for $1\le j\le s$. Then:
		\begin{enumerate}
			\item
			$
			V=U_1\oplus\cdots\oplus U_s.
			$
			
			\item Each vector subspace \(U_j\) is invariant under \(\mathcal{L}\).
			
			\item Any linear map from \(V\) to itself that commutes with \(\mathcal{L}\)
			carries each \(U_j\) into itself.
			
			\item Any vector subspace \(W\) invariant under \(\mathcal{L}\) has the property that
			\[
			W=(W\cap U_1)\oplus\cdots\oplus(W\cap U_s).
			\]
		\end{enumerate}
	\end{lemma}

	Note that the decomposition in 1) of Lemma~\ref{lem-PD} is called the {\em primary decomposition} of $V$ under $\mathcal{L}$, and the vector subspaces $U_j$ are called the {\em primary subspaces} of $V$ under $\mathcal{L}$.
	The following lemma is important to tackle complicated subcases in the proof of the main theorem later. For the sake of readability, its proof is postponed in Appendix \ref{appendix}.

	\begin{lemma}\label{lem-Jor}
		Let $V:=\gf_q^n$, $p:=\Char(\gf_q)$, and let $\tN: V \to V$ be a nilpotent linear transformation whose Jordan type is $(p^t, 1^f)$. For $1 \leq j \leq p$, define $V_j:=V\tN^{j-1}$, $V_{p+1}=\bzero$, and define the map 
		\begin{eqnarray*}
			\tN_{j} : V_j &\rightarrow& V_{j+1},\\
			\bx &\mapsto& \bx\tN.
		\end{eqnarray*} Let $\cC_j:=\cC\tN^{j-1}$ for $1 \leq j \leq p$, and $\cC_{p+1}=\bzero$. Let $K_j:=\Ker(\tN_{j})$, $\nu_j:=\dim(K_j)$, and $g_j:=\dim(\cC_j \cap K_j)$. 
		Let $T$ denote the number of $k$-dimensional subspaces $\cC \subseteq \gf_q^n$ satisfying $\cC\tN \subseteq \cC$. Then
		\[
		T \leq \sum_{\substack{
				\nu_j\ge g_j\ge g_{j+1}\ge 0\\
				\sum g_j=k
		}}
		\prod_{j=1}^p
		q^{g_{j+1}(\nu_j-g_j)}
		\left[\begin{matrix}
			\nu_j-g_{j+1}\\
			g_j-g_{j+1}
		\end{matrix}\right]_q,
		\]
		where $g_{p+1}=0$ and $\nu_j=\dim(K_j)=\left\{\begin{array}{ll}
			t+f, &  j=1, \\
			t, & 2 \leq j \leq p. 
		\end{array}\right.$
	\end{lemma}

	Immediately, we have the following corollary for $p=2$.
	\begin{corollary}\label{cor-ordertwo}
		Let $\tN: \gf_q^n \to \gf_q^n$ be a nonzero linear transformation satisfying
		\[
		\tN^2=\bzero,\
		\dim (\operatorname{Im}(\tN))=t, \textrm{ and }
		\dim(\Ker (\tN))=n-t.
		\]
		Let $T$ denote the number of $k$-dimensional subspaces $\cC \subseteq \gf_q^n$ satisfying $\cC\tN \subseteq \cC$. Then
		\[
		T \leq \sum_{r=0}^{\min\{t,\lfloor k/2\rfloor\}}
		\genfrac{[}{]}{0pt}{}{t}{r}_q
		\genfrac{[}{]}{0pt}{}{n-t-r}{k-2r}_q
		q^{(n-t-k+r)r}.
		\] 
	\end{corollary}

	\section{On the existence of nontrivial automorphism groups}\label{sec3}
	Let $\cC \subseteq \gf_q^n$ be a uniformly random $k$-dimensional linear subspace with a generator matrix $\tG=[\bg_1,\bg_2,\ldots,\bg_n]\in \mathbb F_q^{k \times n}.$ When $m :=\min\{k, n-k\}$ is relatively small, zero or proportional columns are more likely to appear in its generator matrix or parity-check matrix. In this case, the code $\cC$ may have a nontrivial automorphism group. In this section, we study the probability that a random $k$-dimensional linear code $\cC$ has a nontrivial automorphism group when $m$ is upper bounded by a value in terms of $n$ and $q$.
	
	To this end, a key observation is that any $k$-dimensional subspace $\cC \subseteq \gf_q^n$ has a nontrivial automorphism of order $\ell \ge 2$ if $m <\log_q \left(\frac{(q-1)n}{\ell -1} -(q-2)\right)$. By Lemma~\ref{lem-dualag}, we have $\MAut(\cC) \cong \MAut(\cC^{\perp})$, then it suffices to consider the case $k <\log_q \left(\frac{(q-1)n}{\ell -1} -(q-2)\right)$. More precisely, partition finite vector space $\gf_q^k$ into $1+\frac{q^k-1}{q-1}$ classes: the singleton $\{\bzero\}$ and the $1$-dimensional subspaces with the zero vector removed. Since $k <\log_q \left(\frac{(q-1)n}{\ell -1} -(q-2)\right)$, we have $n>(\ell-1)\left(1+\frac{q^k-1}{q-1}\right)$. By the pigeonhole principle, at least $\ell$ columns of the generator matrix $\tG$ must lie in the same class. 
	If these $\ell$ columns are all zero vectors, let $\tP$ be the permutation matrix corresponding to an $\ell$-cycle on the associated coordinates and fixing all other coordinates. Then $\tP\in\MAut(\cC) \setminus \MAut_0$ and $\tP$ has order $\ell$.
	If these $\ell$ columns $\bg_{i_j}$ for $1 \le j \le \ell$ are in one class of $1$-dimensional subspace with $\bzero$ removed, then $\bg_{i_j}=a_j\bv$ for some $\bv\in\mathbb F_q^k \setminus \{\bzero\}$ and $a_1,\ldots,a_\ell\in\mathbb F_q^*$. Let $\be_i \in \gf_q^n$ denote the $i$-th standard basis vector whose $i$-th coordinate is $1$ and all other coordinates are $0$. Define a monomial matrix $\tM$ by $\be_{i_j}\tM=\frac{a_{j+1}}{a_j}\be_{i_{j+1}}$ for $1\le j\le \ell-1$, $\be_{i_\ell}\tM=\frac{a_{1}}{a_\ell}\be_{i_{1}}$ and $\be_i\tM=\be_i$ for all $i \in [n] \setminus \{i_1, \dots, i_{\ell}\}$. It is easy to deduce that $\tM \in \MAut(\cC) \setminus \MAut_0$ and has order $\ell$. 
	
	Following the idea above, we are able to prove the result as below.
	
	\begin{theorem}\label{thm-nontrivial}
		Fix a prime power $q$, an integer $\ell\ge2$, and a real number $\varepsilon>0$. Let $\cC \subseteq \gf_q^n$ be a uniformly random $k$-dimensional linear subspace, where $k=k(n)$ and let $m :=\min\{k, n-k\}$. If
		\begin{eqnarray*}
			m \le
			\frac{\ell}{\ell-1}\log_q n
			+\log_q(q-1)
			-\frac{\ell}{\ell-1}\log_q \ell,
		\end{eqnarray*}
		there exists $n_0=n_0(q,\ell,\varepsilon)$ such that,  for all $n>n_0$, it holds that
			\[\Pr[ \MAut(\cC) \neq \MAut_0] \geq \frac{1}{2}-\varepsilon.\]
	\end{theorem}
	
    \begin{IEEEproof}
		Since $\MAut(\cC) \cong \MAut(\cC^{\perp})$ by Lemma~\ref{lem-dualag}, without loss of generality, we restrict attention to $k \leq n/2$, i.e., $m=k$.
		If $k<\log_q \left(\frac{(q-1)n}{\ell -1} -(q-2)\right)$, the deterministic observation above already proves the assertion. We may therefore assume that $k \ge \log_q \left(\frac{(q-1)n}{\ell -1} -(q-2)\right)$, which in particular implies that $k$ tends to infinity as $n$ tends to infinity. 
		
		Let $\tG=[\bg_1,\bg_2,\dots,\bg_n]\in \mathbb F_q^{k\times n}$ be a uniformly random matrix. Define $b_k:=(q-1)^{\ell-1}q^{-k(\ell-1)}$. Clearly, $b_k$ goes to zero as $k$ goes to infinity. Let $u$ be the smallest integer in terms of $n$ such that $u\ge \ell$ and $B_u:=\binom u\ell b_k\ge 1$, where $\binom u\ell$ is the binomial coefficient. Now we show that $u \leq n$. Since
		\begin{eqnarray*}
			k(\ell-1)
			\le
			\ell\log_q n
			+(\ell-1)\log_q(q-1)
			-\ell\log_q\ell,
		\end{eqnarray*} we have $q^{k(\ell-1)}
		\le
		(q-1)^{\ell-1}\frac{n^\ell}{\ell^\ell}$ and $b_k=(q-1)^{\ell-1}q^{-k(\ell-1)}
		\ge
		\frac{\ell^\ell}{n^\ell}$. On the other hand, for $n \geq \ell$, we have $$\binom n\ell
		=
		\frac{n(n-1)\cdots(n-\ell+1)}{\ell!}
		\ge
		\left(\frac n\ell\right)^\ell \ge \frac{1}{b_k}.$$ It then follows that $\binom n\ell b_k
		\ge
		1$.  Since $u$ is the smallest integer such that $\binom u\ell b_k\ge 1$, we have $u \leq n$. 
		
		Denote by $E_S$ the event that all columns in the multiset $\{\bg_i: i\in S\}$ are nonzero and belong to the same one-dimensional subspace of $\gf_q^k$, where $S$ is an $\ell$-subset of $[u]$. Let random variable $X:=\sum_{S}\mathbf 1_{E_S}$, where $\mathbf 1_{E_S}$ denotes the indicator of the event $E_S$. Then $X$ must be $\geq 0$, and furthermore $X> 0$ means that there are $\ell$ columns in the first $u$ columns that are proportional to each other. In such a case, we are able to construct an automorphism of order $\ell$. Now we estimate $\Pr[X>0]$. For a fixed $\ell$-subset $S$ of $[u]$, since the columns are chosen uniformly at random, we have 
		\begin{equation}\label{eq-PrES}
			\Pr[E_S]
			=
			(1-q^{-k})
			\left(\frac{q-1}{q^k}\right)^{\ell-1}
			=
			(1-q^{-k})b_k.
		\end{equation}
		Hence, the expectation of random variable $X$ is $\Lambda:= \mathbb E [X]
		=
		\binom u\ell (1-q^{-k})b_k$. Recall that $B_u=\binom u\ell b_k$. By the definition of $u$, we have $\binom{u-1}\ell b_k<1$. Then we have $B_u
		=
		\frac{\binom u\ell}{\binom{u-1}\ell}
		\binom{u-1}\ell b_k
		<
		\frac{u}{u-\ell}.$ Note that $u$ tends to infinity as $b_k$ tends to zero. Hence, for any given $\delta > 0$, we have $1 \le B_u\le 1+\delta$ and $1-q^{-k}\ge 1-\delta$ for sufficiently large $n$. Therefore, it follows that 
		\begin{equation}\label{eq-exp}
			1-\delta
			\le
			\Lambda
			\le
			1+\delta.
		\end{equation}
		By the Bonferroni inequality, we obtain
		\begin{eqnarray}\label{eq-lowerbound}
			\Pr[X>0]
			&=&
			\Pr\left[\bigcup_{S}E_S\right] \nonumber\\
			&\ge&
			\sum_{S}\Pr[E_S]
			-
			\sum_{S\ne T}\Pr[E_S\cap E_T] \nonumber\\
			&=& \Lambda-\sum_{S\ne T}\Pr[E_S\cap E_T]
			.
		\end{eqnarray}
		If $S\cap T=\emptyset$, then the events $E_S$ and $E_T$ are independent, and we have $\Pr[E_S\cap E_T]=\Pr[E_S]\Pr[E_T]$.
		Moreover, by Eq. (\ref{eq-PrES}), we have
		\begin{eqnarray}\label{eq-empty}
			\sum_{S\ne T, \ S\cap T=\emptyset}\Pr[E_S\cap E_T]
			=
			\sum_{S\ne T, \ S\cap T=\emptyset}\Pr[E_S]^2 \leq \binom{\binom u\ell}{2}\Pr[E_S]^2
			\le
			\frac12 \Lambda^2.
		\end{eqnarray}
		If $|S\cap T|=s$, where $1\le s\le \ell-1$, then the number of such pairs $\{S,T\}$ is $\frac{1}{2}\binom u\ell \binom\ell s\binom{u-\ell}{\ell-s}$. If events $E_S$ and $E_T$ both occur, then all the columns belonging to $S \cup T$ are non-zero and are contained in the same one-dimensional subspace of $\gf_q^k$. By the principle of inclusion and exclusion, $|S\cup T|=|S|+|T|-|S \cap T|=2\ell-s$ , and then we have 
		$$\Pr[E_S\cap E_T]
		=
		(1-q^{-k})
		\left(\frac{q-1}{q^k}\right)^{2\ell-s-1}.$$
		This yields
		\begin{eqnarray}\label{eq-nonempty1}
			\sum_{\ S\ne T,\ |S\cap T|=s}\Pr[E_S\cap E_T] \leq \frac{1}{2}
			\binom u\ell
			\binom\ell s
			\binom{u-\ell}{\ell-s}
			(1-q^{-k})
			(q-1)^{2\ell-s-1}
			q^{-k(2\ell-s-1)}.
		\end{eqnarray}
		Dividing both sides of inequality (\ref{eq-nonempty1}) by $\Lambda^2$, we obtain 
		\begin{eqnarray}\label{eq-nonempty2}
			\frac{\sum\limits_{ S\ne T,\ |S\cap T|=s}\Pr[E_S\cap E_T]}{\Lambda^2}
			\le
			\frac{\binom\ell s\binom{u-\ell}{\ell-s}}{2\binom u\ell}
			(1-q^{-k})^{-1}
			(q-1)^{1-s}
			q^{k(s-1)}.
		\end{eqnarray}
		Recall that $u$ grows as $n$ grows, and then for sufficiently large $n$, we have $u \ge 2\ell$. Note that
		$
		\binom{u-\ell}{\ell-s}
		\le
		\frac{u^{\ell-s}}{(\ell-s)!}$ and $\binom u\ell
		\ge
		\frac{(u-\ell+1)^\ell}{\ell!}
		\ge
		\frac{(u/2)^\ell}{\ell!}$.
		It then follows that
		\begin{eqnarray}\label{eq-nonempty3}
			\frac{\binom\ell s\binom{u-\ell}{\ell-s}}{\binom u\ell}
			\le
			2^\ell
			\binom\ell s
			\frac{\ell!}{(\ell-s)!}
			u^{-s}.
		\end{eqnarray}
		Since $B_u=\binom u\ell b_k\ge1$, we have $q^{k(\ell-1)}
		\le
		(q-1)^{\ell-1}\binom u\ell
		\le
		(q-1)^{\ell-1}\frac{u^\ell}{\ell!}.$ Hence, we have
		\begin{eqnarray}\label{eq-nonempty4}
			q^{k(s-1)}
			\le
			(q-1)^{s-1}
			(\ell!)^{-\frac{s-1}{\ell-1}}
			u^{\frac{\ell(s-1)}{\ell-1}}.
		\end{eqnarray}
		Furthermore, $(1-q^{-k})^{-1}\le \frac{q}{q-1}$ as $k \geq 1$. Substituting inequalities (\ref{eq-nonempty3}) and (\ref{eq-nonempty4}) into (\ref{eq-nonempty2}), we obtain
		\begin{eqnarray*}
			\frac{\sum\limits_{\ S\ne T,\ |S\cap T|=s}\Pr[E_S\cap E_T]}{\Lambda^2}
			\le
			C_{\ell,q,s}
			u^{\frac{s-\ell}{\ell-1}},
		\end{eqnarray*}
		where $C_{\ell,q,s}
		=
		\frac{q}{q-1}
		2^{\ell-1}
		\binom\ell s
		\frac{\ell!}{(\ell-s)!}
		(\ell!)^{-\frac{s-1}{\ell-1}}$ is a constant depending only on $\ell$, $q$, and $s$. Since $1\le s\le \ell-1$, we have $\frac{s-\ell}{\ell-1}<0$. For any given $\delta>0$, since $u$ grows as $n$ grows, for all sufficiently large $n$, we have
		\begin{eqnarray}\label{eq-nonempty}
			\sum_{s=1}^{\ell-1}\sum\limits_{ S\ne T,\ |S\cap T|=s}\Pr[E_S\cap E_T] \le \delta \Lambda^2. 
		\end{eqnarray}
		Hence, by inequalities (\ref{eq-lowerbound}), (\ref{eq-empty}) and (\ref{eq-nonempty}), we have for any given $\delta>0$, there exists $n_1$ such that $$\Pr[X>0]
		\ge
		\Lambda-\frac12 \Lambda^2-\delta \Lambda^2$$ for all $n > n_1$.
		By inequality (\ref{eq-exp}), we obtain $\Lambda-\frac12 \Lambda^2
		=
		\frac12-\frac{(\Lambda-1)^2}{2}
		\ge
		\frac12-\frac{\delta^2}{2}$ and $\delta \Lambda^2\le \delta(1+\delta)^2$. For any given $\varepsilon>0$, choose $\delta > 0$ such that $\frac{\delta^2}{2}+\delta(1+\delta)^2<\frac{\varepsilon}{2}$. Then there exists $n_1$ such that $\Pr[X>0] \ge \frac12-\frac{\varepsilon}{2}$ for all $n > n_1$. 
		
		Note that the probability that a random $k \times n$ matrix $\tG$ has full rank is given by $\Pr[\text{rank} (\tG) =k]
		=
		\prod_{j=0}^{k-1}(1-q^{j-n}).$ Since $\prod_{j=0}^{k-1}(1-q^{j-n})
		\ge
		1-\sum_{j=0}^{k-1}q^{j-n}$, we have $\Pr[\text{rank} (\tG) \neq k]
		\le
		\sum_{j=0}^{k-1}q^{j-n}
		=
		\frac{q^k-1}{(q-1)q^n}$. Since $k=O(\log n)$, we have $q^k/q^n$ goes to zero as $n$ goes to infinity. It then follows that for any given $\varepsilon>0$, there exists $n_2$ such that $\Pr[\text{rank} (\tG) \neq k]\le \varepsilon/2$ for all $n > n_2$. Under the condition $\text{Rank} (\tG) =k$, $\langle \tG \rangle$ is uniformly distributed over $\Gr_q(k,n)$.  Hence, we have
		\begin{eqnarray*}
			\lefteqn{\Pr[X>0 \mbox{ and } \text{rank} (\tG) =k]}\\&=&\Pr[X> 0]-\Pr[X>0 \mbox{ and } \text{rank} (\tG) \neq k]\\
			&\ge&
			\Pr[X>0]-\Pr[\text{rank} (\tG) \neq k]
			\ge
			\frac12-\varepsilon.
		\end{eqnarray*}
		Therefore, for a uniformly random $k$-dimensional linear subspace $\cC$, for fixed $\ell$, and for any given $\varepsilon>0$, there exists $n_0=\max\{n_1, n_2\}$ such that
		\begin{eqnarray*}
			\lefteqn{\Pr[\MAut(\cC) \neq \MAut_0]} \\ &\geq& \Pr[ \MAut(\cC)/\MAut_0 \text{ contains an element of order $\ell$}] \\
			&=& \Pr[\MAut(\langle\tG\rangle)/\MAut_0 \text{ contains an element of order $\ell$} ~|~ \text{rank} (\tG) =k] \\
			&\ge& \Pr[X>0 ~|~ \text{rank} (\tG) =k]\\
			&\ge& \Pr[X>0 \mbox{ and } \text{rank} (\tG) =k]
			\ge \frac12-\varepsilon
		\end{eqnarray*}
		for all $n > n_0$. The proof is then completed.
	\end{IEEEproof}
	
	By letting $\ell=2$, we obtain the following corollary immediately.
	\begin{corollary}
		Fix a prime power $q$ and a real number $\varepsilon>0$. Let $\cC \subseteq \gf_q^n$ be a uniformly random $k$-dimensional linear subspace, where $k=k(n)$, and let $m :=\min\{k, n-k\}$. If
		$
		m \le
		2\log_q n
		+\log_q(q-1)
		-2\log_q 2,
		$
		then there exists $n_0=n_0(q,\varepsilon)$ such that for every $n>n_0$, it holds that
			\[\Pr[\MAut(\cC) \neq \MAut_0] \geq \frac{1}{2}-\varepsilon.\]
	\end{corollary}
	
	\section{Proof of the Main Theorem}\label{sec4}
	\subsection{Overview of the proof}
	We recall the main theorem as follows.
	\begin{theorem}\label{main-thm}
		Fix a prime power $q$ and a real number $\varepsilon>0$. Let $\cC \subseteq \gf_q^n$ be a uniformly random $k$-dimensional linear code, and let $m :=\min\{k, n-k\}$. If $m\ge (2+\varepsilon)\log_q n$,
		then $\cC$ has a trivial automorphism group with probability tending to one as $n$ goes to infinity.
	\end{theorem}
	
    Let $\cC$ be chosen uniformly at random from $\Gr_q(k, n)$ and $\MAut_0:=\mathbb F_q^*\tI_n$. Define the quotient group $H_{\cC}:=\MAut(\cC) /\MAut_0$ and the quotient map $\phi: \MAut(\cC) \rightarrow H_{\cC}$ with $\phi(\tM) = \mathbb F_q^*\tM$. By Lemma~\ref{lem-cauchy}, if the automorphism group of $\cC$ is nontrivial, i.e., $\MAut_0 \subsetneq \MAut(\cC)$, then there must exist $\tM \in \MAut(\cC) \setminus \MAut_0$ whose image $\overline{\tM}:=\phi(\tM) \in H_{\cC}$ has prime order.
	Recall that the monomial group $M_{n,q}$ is defined as
	\[
	M_{n,q}
	:= (\gf_q^{*})^n \rtimes S_n =
	\{\tD\tP: \tD=\operatorname{diag}(\alpha_1,\dots,\alpha_n), \alpha_i\in\mathbb F_q^*, \tP\in S_n\}.
	\]
	Define the quotient group $H:=M_{n,q}/\MAut_0$. Note that for each code $\cC$, the group $H_{\cC}$ is naturally embedded in $H$. 
	
	Let $E_0$ denote the event that $\cC$ has a nontrivial automorphism group, i.e., $\MAut_0$ is a {\em proper} subgroup of $\MAut(\cC)$. In the following, we will study the asymptotic behavior of $\Pr[E_0]$. In other words, $\Pr[E_0]$ tends to zero if $\cC$ has a trivial automorphism group with high probability.  Let $E_1$ denote the event that there exists a nonidentity element $\overline \tM\in H$ of prime order such that $\cC$ is $\overline \tM$-invariant, i.e., \[\cC \tM=\cC \quad \text{ for one representative $\tM$ of } \overline \tM.\]
In fact, this is equivalent to saying that $\cC \tM=\cC$ for {\em every} representative $\tM$ of $\overline \tM$, since two representatives differ by a nonzero scalar, which stabilizes every linear code.
Then the occurrence of $E_0$ implies the occurrence of $E_1$, and hence, $\Pr[E_0] \leq \Pr[E_1]$. 
	It then follows from the union bound that
	\begin{eqnarray}\label{eq-union}
		\Pr\left[
		E_0
		\right] \leq \Pr[E_1]
		\le
		\sum_{\substack{
				1 \neq \overline \tM\in H,\\ \ord(\overline \tM)\text{ prime}
		}}
		\Pr\left[ \cC\tM=\cC\right].
	\end{eqnarray}
	Thus it remains to estimate the probabilities appearing on the right-hand side of inequality (\ref{eq-union}). Note that $$\Pr\left[ \cC\tM=\cC\right]=\frac{| \{\cC \in \Gr_q(k,n): \cC\tM=\cC\}|}{|\Gr_q(k, n)|},$$ where $|\Gr_q(k, n)|=\genfrac{[}{]}{0pt}{}{n}{k}_q$ is the Gaussian coefficient. It then suffices to count the number of $\tM$-invariant $k$-dimensional subspaces $\cC$. The primary decomposition (in Lemma \ref{lem-PD}) is a useful tool for dealing with such counting problems.
	
	Let $1\neq \overline \tM\in H$ and suppose that $\operatorname{ord}(\overline \tM)=\ell$, where $\ell$ is prime. Fix a representative $\tM=\tD\tP\in M_{n,q}$ of coset $\overline \tM$, where $\tD=\operatorname{diag}(\alpha_1,\dots,\alpha_n)$ with $\alpha_i\in \mathbb F_q^*$, and $\tP \in S_n$. Since $\operatorname{ord}(\overline \tM)=\ell$, we have $\overline \tM ^\ell= \gf_q^* \tI_n$, and then there exists $\lambda \in \gf_q^*$ such that $\tM^{\ell}=(\tD\tP)^{\ell}=\lambda\tI_n$. It follows that $\tP^{\ell}=\tI_n$. This leads to the following two cases: 
	\begin{enumerate}
		\renewcommand{\labelenumi}{(\roman{enumi})}
		\item $\tP = \tI_n$;
		\item $\tP \neq \tI_n$.
	\end{enumerate}
	Next we separately discuss the basic ideas for dealing with the two cases, and in the following subsections prove them in detail, respectively.
	
	\paragraph*{(i) $\tP = \tI_n$}
	In this case, we need only consider the action of diagonal matrix $\tD=\operatorname{diag}(\alpha_1,\ldots,\alpha_n)$. By the eigenspace decomposition of $\tD$, $\mathbb F_q^n=V_1\oplus\cdots\oplus V_s$, and then every $\tD$-invariant subspace $\cC$ decomposes as $\cC=(\cC\cap V_1)\oplus\cdots\oplus(\cC\cap V_s).$ Each direct summand $\cC\cap V_j$ is a subspace of $V_j$ with dimension $r_j$. Now we are able to estimate the number of $\tD$-invariant $k$-dimensional subspaces $\cC$ by counting the possible choices of $\cC\cap V_j$ in $V_j$ for each tuple
	$(r_1,\ldots,r_s)$ with $\sum_{j=1}^s r_j=k$. For details of the proof in this case, see Subsection~\ref{subsec4.1}.
	
	\paragraph*{(ii) $\tP \neq \tI_n$} Note that in this case, since $\tP^{\ell}=\tI_n$ and $\ell$ is a prime, all nontrivial cycles of $\tP$ have length $\ell$. Write $\tP=\gamma_1\gamma_2\cdots \gamma_t$, where $t\ge 1$, each $\gamma_j$ is a cycle of length $\ell$, and the remaining points are fixed by $\tP$. Let
	$
	W_j:=\langle \be_{i} : i \in \gamma_j \rangle_{\mathbb F_q}.
	$ By analyzing the action of $\tM$ on $W_j$, we deduce that the minimal polynomial of $\tM$ is $X^{\ell}-\lambda$. 
	To this end, we discuss separately the following two subcases:
	(a) $\ell \neq \Char(\gf_q)$; (b) $\ell = \Char(\gf_q)$.
	
	\noindent (a) $\ell \neq \Char(\gf_q)$. In this subcase, $X^{\ell}-\lambda$ has no repeated roots. Hence, it can be factored as a product of distinct monic irreducible polynomials over $\gf_q$.
	Let $X^\ell-\lambda=f_1(X)\cdots f_s(X)$ be the irreducible factorization of $X^{\ell}-\lambda$ in $\gf_q[X]$. 
	By Lemma \ref{lem-PD}, we have $\mathbb F_q^n= V_1\oplus\cdots\oplus V_s,$ where $V_i=\Ker(f_i(\tM))$. Moreover, any $\tM$-invariant subspace $\cC$ can be decomposed as $\cC=(\cC\cap V_1)\oplus\cdots\oplus(\cC\cap V_s).$ We are then able to estimate the number of $\tM$-invariant $k$-dimensional subspaces $\cC$ by counting the possible choices of $\cC\cap V_i$ in $V_i$. 
	
	\noindent (b) $\ell = \Char(\gf_q)=p$. In this subcase, we first choose a suitable representative $\tM$ of coset $\overline{\tM}$ such that the minimal polynomial of $\tM$ is $X^{\ell}-1=X^p-1=(X-1)^p$ . However, the primary decomposition of $\gf_q^n$ under $\tM$ is trivial in this situation, and therefore estimating the number of $\tM$-invariant subspaces $\cC$ by counting the possible choices of their intersections with the primary subspaces of $\gf_q^n$ under $\tM$ is insufficient to obtain the desired result. Instead, set $\tN:=\tM-\tI_n$. By analyzing the action of $\tN$ on $W_j$, we deduce that $\tN$ is nilpotent and has Jordan type $(p^t, 1^f)$. Since every $\tM$-invariant subspace $\cC$ satisfies $\cC\tN \subseteq \cC$, we construct the chain $\cC_p\subseteq \cC_{p-1}\subseteq\cdots\subseteq \cC_1$ by defining $\cC_j=\cC\tN^{j-1}$ for $1 \leq j \leq p$, and $\cC_{p+1}=\bzero$. Then we estimate the number of $\tM$-invariant $k$-dimensional subspaces $\cC$ by counting the possible choices of lifts from $\cC_{j+1}$ to $\cC_j$ for all $j$. For details of the proof in this case, see Subsection~\ref{subsec4.2}. 
	
	After estimating the number of $\tM$-invariant $k$-dimensional subspaces $\cC$, we sum these estimates over all relevant nonidentity prime-order elements $\overline{\tM}$ to obtain upper bounds for the corresponding sums of probabilities. The proof of the main theorem is then completed by analyzing the asymptotic behavior of these bounds. 

	\subsection{$\tP = \tI_n$.}\label{subsec4.1}
	Following the idea we discuss above, we are able to give the results as below.
	
	\begin{theorem}\label{thm-I}
		Fix a prime power $q$ and a real number $\varepsilon>0$. Let $\cC \subseteq \gf_q^n$ be a uniformly random $k$-dimensional linear subspace, where $k=k(n)$, let $m:=\min\{k, n-k\}$, and let $H:=M_{n,q}/\MAut_0$. Define $\mathcal P_q:=\{\ell:\ell\ \text{prime},\ \ell\mid(q-1)\}.$ If $m\ge (2+\varepsilon)\log_q n$, then the sum of the probabilities that $\cC$ is invariant under a diagonal matrix $\tM=\tD$, whose coset $\overline{\tM}$ in $H$ has prime order $\ell$, tends to zero as $n$ goes to infinity. More precisely, 
		\[\sum_{\ell \in \mathcal{P}_q}
		\sum_{\substack{
				1\neq \overline \tM\in H\\
				\operatorname{ord}(\overline \tM)=\ell\\
				\tP = \tI_n
		}}
		\Pr\left[ \cC\tD=\cC\right]
		\] tends to zero as $n$ goes to infinity.
	\end{theorem}
	
	\begin{IEEEproof}
        Without loss of generality, assume that $k \leq n/2$ and $m=k$. Let $1\neq \overline \tM\in H$ with $\operatorname{ord}(\overline \tM)=\ell$, where $\ell$ is prime. Assume that $\tP=\tI_n$, then there exists $\lambda \in \gf_q^*$ such that $\tD^{\ell}=\lambda\tI_n$, i.e., $a_i^{\ell}=\lambda$ for $i=1,2,  \cdots, n$. Let $a_{i_0}$ be the element that appears most frequently on the diagonal of $\tD$. Replace $\tD$ by the representative $a_{i_0}^{-1}\tD$. Then we have $\tD^\ell=\tI_n$, $1$ is an entry of maximum multiplicity in $\tD$, and $\alpha_i^\ell=1$ for every element $\alpha_i$ on the diagonal of $\tD$. Since $\overline{\tD}\neq1$, at least one $\alpha_i$ is a nontrivial $\ell$-th root of unity, and hence, $\ell\mid(q-1)$.
		
		Let $\xi_1,\ldots,\xi_s\in \mathbb F_q^\ast$ be all distinct eigenvalues of $\tD$, where $2\le s\le \ell\le q-1.$ Then the minimal polynomial of $\tD$ is $\prod_{j=1}^s(X-\xi_j)$. Note that the polynomials $X-\xi_j$ are pairwise coprime. Let $V_j=\Ker(\tD-\xi_j \tI_n)$ for $1\le j\le s.$ By Lemma \ref{lem-PD}, $\gf_q^n$ can be decomposed as a direct sum of the eigenspaces of $\tD$, i.e., $\mathbb F_q^n=V_1\oplus\cdots\oplus V_s$. Moreover, every $\tD$-invariant subspace $\cC$ decomposes as the direct sum of its intersections with these eigenspaces, i.e., $\cC=(\cC\cap V_1)\oplus\cdots\oplus(\cC\cap V_s).$
		Denote $n_j=\dim V_j.$ Without loss
		of generality, assume that $n_1=\max_{1\le j\le s}n_j.$ Define $\bar{n}:=n-n_1$ 
		and $r_j:=\dim (\cC\cap V_j).$ Note that $0\le r_j\le n_j$ and $\ r_1+\cdots+r_s=k.$ Hence, we have
		\begin{eqnarray*}
			|\{\cC \in \text{Gr}_q(k, n) : \cC\tD=\cC\}| \leq
			\sum_{\substack{r_j \ge 0 \\ \sum r_j=k}}
			\prod_{j=1}^s
			\genfrac{[}{]}{0pt}{}{n_j}{r_j}_q .
		\end{eqnarray*}
		By Lemma \ref{lem-Gau}, there exist constants $c_1$, $c_2$ such that
		\begin{eqnarray*}
			\genfrac{[}{]}{0pt}{}{n_j}{r_j}_q \leq c_2 q^{r_j(n_j-r_j)} \mbox{ and } \genfrac{[}{]}{0pt}{} nk_q\ge c_1 q^{k(n-k)}.
		\end{eqnarray*}
		Therefore, we obtain
		\begin{eqnarray*}
			\Pr[\cC\tD=\cC]
			\le
			\frac{c_2^s}{c_1}
			\sum_{\substack{r_j \ge 0\\ \sum r_j=k}}
			q^{\sum r_j(n_j-r_j)-k(n-k)}.
		\end{eqnarray*}
		Since $s\le q-1$, the number of $(r_1, r_2, \cdots, r_s)$ such that $r_1+\cdots+r_s=k$ is at most $(k+1)^{q-1}$. Note that $n_j\le n_1=n-\bar{n}$ for all $1 \le j \le s$. Then we have
		\begin{eqnarray*}
			\sum_{j=1}^s r_j(n_j-r_j)
			\le
			\sum_{j=1}^s r_j n_j
			\le
			(n-\bar{n})\sum_{j=1}^s r_j
			=
			k(n-\bar{n}).
		\end{eqnarray*}
		It then follows that
		\begin{eqnarray}\label{eq-1}
			\Pr[ \cC\tD=\cC]
			\le
			\frac{c_2^s}{c_1} (k+1)^{q-1}q^{-k(\bar{n}-k)}.
		\end{eqnarray}
		Note that this estimate is strong when $\bar{n}>k$, but it is weak when $\bar{n} \leq k$. We therefore derive another estimate in what follows.
		
		Let $U:=V_1$ and $W:=V_2\oplus\cdots\oplus V_s$. Moreover, we have $\cC=(\cC\cap U)\oplus(\cC\cap W)$ if $\cC\tD=\cC$. Define $r:=\dim(\cC\cap W).$ Then we have $0\le r\le \min\{k,\bar{n}\}$ and
		$\dim(\cC\cap U)=k-r.$ Therefore, we obtain
		\begin{eqnarray*}
			|\{\cC \in \text{Gr}_q(k, n) : \cC\tD=\cC\}| \leq \sum_{r=0}^{\min\{k,\bar{n}\}}
			\genfrac{[}{]}{0pt}{} {\bar{n}}{r}_q
			\genfrac{[}{]}{0pt}{}{n-\bar{n}}{k-r}_q.
		\end{eqnarray*}
		By Lemma \ref{lem-Gau}, we have
		\begin{eqnarray}\label{eq-2}
			\Pr[\cC\tD=\cC]
			\le
			\frac{c_2^2}{c_1}
			\sum_{r=0}^{\min\{k,\bar{n}\}}
			q^{-\Delta(r, \bar{n})},
		\end{eqnarray}
		where $\Delta(r,\bar{n})
		:=
		k(n-k)-r(\bar{n}-r)-(k-r)(n-\bar{n}-k+r)
		=
		k\bar{n}+r(n-2k-2\bar{n})+2r^2.
		$
		
		If $k+\bar{n}\le n/2$, then for every $r$,
		$
		\Delta(r,\bar{n})=k\bar{n}+r(n-2k-2\bar{n})+2r^2 \geq k\bar{n}.
		$
		It then follows that
		\[
		\sum_{r=0}^{\min\{k,\bar{n}\}}
		q^{-\Delta(r,\bar{n})} \leq 
		(k+1)q^{-k\bar{n}},
		\]
		and 
		\begin{eqnarray}\label{eq-le}
			\Pr[\cC\tD=\cC]
			\le
			\frac{c_2^2}{c_1}(k+1)q^{-k\bar{n}}.
		\end{eqnarray}
		
		Let $\theta:=\log_q n$ and $\eta:=\log_q(q-1).$ If $k+\bar{n} > n/2$ and $k \leq \theta^2$, then for sufficiently large $n$, we have $k \leq \frac{n}{6}$ and $\bar{n}>\frac n2-k \geq \frac{n}{3}$. In this case, we have
		\begin{eqnarray}\label{eq-ge-le}
			\Pr[ \cC\tD=\cC]
			\le
			\frac{c_2^s}{c_1} (k+1)^{q-1}q^{-k(\bar{n}-k)}.
		\end{eqnarray}
		
		If $k+\bar{n} > n/2$ and $k > \theta^2$, then for every $r$,
		$\Delta(r,\bar{n})
		\ge
		k\bar{n}-\frac{(2k+2\bar{n}-n)^2}{8}.$
		Since $k\le n/2$, we have $2k+2\bar{n}-n\le 2\bar{n}$ and
		\begin{eqnarray*}
			\Delta(r,\bar{n})
			\ge
			k\bar{n}-\frac{\bar{n}^2}{2}.
		\end{eqnarray*} 
		If $k+\bar{n} > n/2$, $k > \theta^2$ and $\bar{n} \leq k$, then we have $\Delta(r,\bar{n}) \ge k\bar{n}-\frac{\bar{n}^2}{2} \ge \frac{k\bar{n}}{2}$. By inequality (\ref{eq-2}), we have
		\begin{eqnarray}\label{eq-ge-k}
			\Pr[ \cC\tD=\cC] \le \frac{c_2^2}{c_1}(k+1)q^{-k\bar{n}/2}.
		\end{eqnarray}
		If $k+\bar{n} > n/2$, $k > \theta^2$ and $\bar{n} \ge 2k$, then we have $k(\bar{n}-k)\ge \frac{k\bar{n}}{2}$. By inequality (\ref{eq-1}), we have
		\begin{eqnarray}\label{eq-ge-2k}
			\Pr[ \cC\tD=\cC]
			\le
			\frac{c_2^s}{c_1} (k+1)^{q-1}q^{-k\bar{n}/2}.
		\end{eqnarray}
		If $k+\bar{n}>n/2$, $k > \theta^2$ and $k < \bar{n}< 2k$, consider $k\bar{n}-\frac{\bar{n}^2}{2}$ and $k(\bar{n}-k)$. Let $x=\frac {\bar{n}}{k}\in(1,2)$. Then $k\bar{n}-\frac{\bar{n}^2}{2} = k^2\left(x-\frac{x^2}{2}\right)$ and $k(\bar{n}-k)=k^2(x-1)$. It is easy to deduce that $\max\left\{
		x-\frac{x^2}{2},
		x-1
		\right\}
		\ge
		\frac14 x$ for
		$1\le x\le 2.$ Then we obtain $\max\left\{
		k\bar{n}-\frac{\bar{n}^2}{2},
		k(\bar{n}-k)
		\right\}
		\ge
		\frac14 k\bar{n}$ for $k < \bar{n}< 2k.$ By inequalities (\ref{eq-1}) and (\ref{eq-2}), we have
		\begin{eqnarray}\label{eq-ge-m}
			\Pr[ \cC\tD=\cC]
			\le \frac{c_2^s}{c_1} (k+1)^{q-1}q^{-k\bar{n}/4}.
		\end{eqnarray}
		
		It follows from inequalities (\ref{eq-le}), (\ref{eq-ge-le}), (\ref{eq-ge-k}), (\ref{eq-ge-2k}), and (\ref{eq-ge-m}) that 
		$\Pr[ \cC\tD=\cC]
		\le p_{\bar{n}}$,
		where 
		\begin{eqnarray}\label{eq-pbarn}
			p_{\bar{n}}=
			\begin{cases}
				\frac{c_2^2}{c_1}(k+1)q^{-k\bar{n}},
				& k+\bar{n}\le n/2,\\
				\frac{c_2^s}{c_1}(k+1)^{q-1}q^{-k(\bar{n}-k)},
				& k+\bar{n}>n/2,\ k\le \theta^2,\\
				\frac{c_2^s}{c_1}(k+1)^{q-1}q^{-k\bar{n}/4},
				& k+\bar{n}>n/2,\ k>\theta^2.
			\end{cases}
		\end{eqnarray}
		
		We next sum the preceding estimates for $\Pr[ \cC\tD=\cC]$ over all nontrivial cosets $\overline{\tM} \in H$ of order $\ell$ with $\tP=\tI_n$, where for each coset we use the normalized diagonal representative $\tD$ chosen in the first paragraph of the proof. For fixed $\ell \in \mathcal P_q$ and $\bar{n}$, the number of normalized diagonal representatives $\tD$ is at most $\binom {n}{\bar{n}}(\ell-1)^{\bar{n}}$. Then we have
		\begin{eqnarray*}
			\lefteqn{\sum_{\ell \in \mathcal{P}_q}
				\sum_{\substack{
						1\neq \overline \tM\in H\\
						\operatorname{ord}(\overline \tM)=\ell\\
						\tP = \tI_n
				}}
				\Pr\left[ \cC\tD=\cC\right] }\\
			&\leq&  \sum_{\ell \in \mathcal{P}_q} \sum_{\bar{n}=1}^{n-1} \binom{n}{\bar{n}} (\ell -1)^{\bar{n}} p_{\bar{n}} \\
			&=&  \sum_{\bar{n}=1}^{n-1} \binom{n}{\bar{n}} \left(\sum_{\ell \in \mathcal{P}_q} (\ell -1)^{\bar{n}} \right)p_{\bar{n}} \\
			& \leq & \sum_{\bar{n}=1}^{n-1} \binom{n}{\bar{n}} \left(\sum_{\ell \in \mathcal{P}_q} (\ell -1) \right)^{\bar{n}} p_{\bar{n}} \\
			& \leq & \sum_{\bar{n}=1}^{n-1} \binom{n}{\bar{n}} (q-2)^{\bar{n}} p_{\bar{n}} \\
			& \leq & \sum_{\bar{n}=1}^{n-1} n^{\bar{n}} (q-1)^{\bar{n}} p_{\bar{n}}, 
		\end{eqnarray*}
		where the third inequality holds as $$q-2 \geq |\{\zeta \in \gf_q^*\setminus \{1\}: \ord(\zeta) \mbox{ is prime}\}|=\sum_{\ell \in \mathcal{P}_q} \phi(\ell)=\sum_{\ell \in \mathcal{P}_q} (\ell -1).$$
		
		Recall that $\theta:=\log_q n$ and $\eta:=\log_q(q-1).$ It follows that $n^{\bar{n}}(q-1)^{\bar{n}}=q^{\bar{n}(\theta+\eta)}.$ Then it is sufficient to prove that $\sum\limits_{\bar{n}=1}^{n-1} q^{\bar{n}(\theta+\eta)}p_{\bar{n}}$ goes to zero as $n$ goes to infinity. For $1 \leq \bar{n} \leq n-1$, by Eq. (\ref{eq-pbarn}), the value of $p_{\bar{n}}$ differs between the two cases $k+\bar{n}\le n/2$ and $k+\bar{n}> n/2$.
		Then we consider the following two parts:
		\begin{eqnarray*}
			S_1=
			\sum_{\substack{1\le \bar{n}\le n-1\\ k+\bar{n}\le n/2}}
			q^{\bar{n}(\theta+\eta)}p_{\bar{n}}, \mbox{ } S_2=
			\sum_{\substack{1\le \bar{n}\le n-1\\ k+\bar{n}> n/2}}
			q^{\bar{n}(\theta+\eta)}p_{\bar{n}}.
		\end{eqnarray*}
		
		When $k+\bar{n}\le n/2$, it follows from Eq. (\ref{eq-pbarn}) that $$S_1
		\le
		\frac{c_2^2}{c_1}(k+1)
		\sum_{\bar{n}\ge 1} q^{\bar{n}(\theta+\eta-k)}.$$ Since $q$ is fixed and $\eta=\log_q(q-1)$ is a constant depending only on fixed $q$, and hence $\eta\le \frac{\varepsilon}{4}\theta$ for sufficiently large $n$. Moreover, $k\ge (2+\varepsilon)\theta$. It follows that $$\theta+\eta-k
		\le
		\theta+\frac{\varepsilon}{4}\theta-(2+\varepsilon)\theta
		=
		-\left(1+\frac{3\varepsilon}{4}\right)\theta.$$
		Hence, we have $$S_1
		\le
		\frac{c_2^2}{c_1}(k+1)
		\sum_{\bar{n}\ge 1}
		q^{-(1+\frac{3\varepsilon}{4})\bar{n}\theta}=\frac{c_2^2}{c_1}(k+1)
		\sum_{\bar{n}\ge 1} 
		n^{-(1+\frac{3\varepsilon}{4})\bar{n}} = \frac{c_2^2}{c_1}(k+1)\cdot \frac{n^{-(1+\frac{3\varepsilon}{4})}}{1-n^{-(1+\frac{3\varepsilon}{4})}}
		\le
		\frac{2c_2^2}{c_1}\left(\frac n2+1\right)n^{-(1+\frac{3\varepsilon}{4})},
		$$ which goes to zero as $n$ goes to infinity.
		
		We next consider the case where $k+\bar{n} > n/2$. If $k+\bar{n}>n/2$ and $k>\theta^2$, it follows from Eq. (\ref{eq-pbarn}) that $$S_2
		\le
		\frac{c_2^s}{c_1} (k+1)^{q-1}
		\sum_{\bar{n}\ge 1}
		q^{\bar{n}(\theta+\eta-\frac{k}{4})}.$$ Since $k>\theta^2$, we have $\theta+\eta\le k/8$ for sufficiently large $n$. Then we have $\theta+\eta-\frac{k}{4}
		\le
		-\frac{k}{8}$. Therefore, we obtain $$S_2
		\le
		\frac{c_2^s}{c_1}(k+1)^{q-1}
		\sum_{\bar{n}\ge 1}q^{-k\bar{n}/8} \le 
		\frac{c_2^s}{c_1}(k+1)^{q-1}
		\frac{q^{-k/8}}{1-q^{-k/8}} = \frac{c_2^s}{c_1} \frac{(k+1)^{q-1}}{q^{k/8}-1}, $$ which tends to zero as $n$ tends to infinity, since $k>\theta^2$ implies that $k$ tends to infinity. 
		If $k+\bar{n}>n/2$ and $k\le \theta^2$, we have $\bar{n}>\frac n2-k \geq \frac{n}{3}$ for sufficiently large $n$. It then follows from Eq. (\ref{eq-pbarn}) that
		\begin{eqnarray*}
			S_2
			\le
			\frac{c_2^s}{c_1} (k+1)^{q-1}
			\sum_{\bar{n}\ge n/3}
			q^{\bar{n}(\theta+\eta)-k(\bar{n}-k)}.
		\end{eqnarray*} 
		For sufficiently large $n$, we have $k^2 \leq \varepsilon \bar{n}\theta/4$ and $\eta\le \varepsilon \theta/4$. Note that $k\ge (2+\varepsilon)\theta$. Then we have $k(\bar{n}-k)=k\bar{n}-k^2 
		\ge
		\left(2+\frac{3\varepsilon}{4}\right)\bar{n}\theta.$ It follows that $$
		\bar{n}(\theta+\eta)-k(\bar{n}-k)
		\le
		\bar{n}\left(1+\frac{\varepsilon}{4}\right)\theta
		-
		\bar{n}\left(2+\frac{3\varepsilon}{4}\right)\theta
		=
		-\bar{n}\left(1+\frac{\varepsilon}{2}\right)\theta.
		$$
		Therefore, we have $$S_2
		\le
		\frac{c_2^s}{c_1}(k+1)^{q-1}
		\sum_{\bar{n}\ge n/3}
		q^{-(1+\frac{\varepsilon}{2})\bar{n}\theta}.$$ Moreover, since $k \le \theta^2$ and $\bar{n}<n$, we obtain
		$
		(k+1)^{q-1}\le (\theta^2+1)^{q-1}$ and $\sum\limits_{\bar{n}\ge n/3}
		q^{-(1+\frac{\varepsilon}{2})\bar{n}\theta}
		\le
		nq^{-\frac{(1+\frac{\varepsilon}{2})n\theta}{3}}$. 
		Then we conclude that $$S_2
		\le
		\frac{c_2^s}{c_1}(\theta^2+1)^{q-1}
		nq^{-\frac{(1+\frac{\varepsilon}{2})n\theta}{3}}=\frac{c_2^s}{c_1}(\theta^2+1)^{q-1}
		n^{-\frac{(1+\frac{\varepsilon}{2})n}{3}+1},$$ which tends to zero as $n$ tends to infinity.
		
		By the above discussions, we obtain
		\begin{eqnarray*}
			\sum_{\ell \in \mathcal{P}_q}
			\sum_{\substack{
					1\neq \overline \tM\in H\\
					\operatorname{ord}(\overline \tM)=\ell\\
					\tP = \tI_n
			}}
			\Pr\left[\cC\tD=\cC\right] \leq S_1+S_2,
		\end{eqnarray*}
		which again goes to zero as $n$ goes to infinity.
	\end{IEEEproof}
	
	\subsection{$\tP \neq \tI_n$.}\label{subsec4.2}
	
	Since $\ell$ is prime and $\tP \neq \tI_n$, all nontrivial cycles of $\tP$ have length $\ell$. Thus we may write $\tP=\gamma_1\gamma_2\cdots \gamma_t$ which is a product of disjoint cycles, where $t\ge 1$, each $\gamma_j$ is a cycle of length $\ell$, and the remaining coordinates are fixed points. Write the $j$-th cycle as $\gamma_j=(i_{j,1}\ i_{j,2}\ \cdots\ i_{j,\ell})$. Let $\be_i$ denote the $i$-th standard basis vector whose $i$-th entry is $1$ and all other entries are $0$. Let
	\[
	W_j:=\langle \be_{i_{j,1}},\be_{i_{j,2}},\dots,\be_{i_{j,\ell}}\rangle_{\mathbb F_q}.
	\]
	Then $W_j$ is invariant under the action of $\tM$, and there exist nonzero elements $b_{j,1},\dots,b_{j,\ell}\in \mathbb F_q^*$ such that
	\begin{eqnarray*}
		\be_{i_{j,1}}\tM &=& b_{j,1}\be_{i_{j,2}},\\
		\be_{i_{j,2}}\tM &=& b_{j,2}\be_{i_{j,3}},\\
		&\cdots& \\
		\be_{i_{j,\ell-1}}\tM &=& b_{j,\ell-1}\be_{i_{j,\ell}},\\
		\be_{i_{j,\ell}}\tM &=& b_{j,\ell}\be_{i_{j,1}}.
	\end{eqnarray*}
	Applying $\tM^{\ell}$ to $\be_{i_{j,1}}$ gives
	\[
	\be_{i_{j,1}}\tM^\ell=(b_{j,1}b_{j,2}\cdots b_{j,\ell})\be_{i_{j,1}}.
	\]
	On the other hand, $\tM^\ell=\lambda \tI_n$ gives $\be_{i_{j,1}}\tM^\ell=\lambda \be_{i_{j,1}}$. Therefore $b_{j,1}b_{j,2}\cdots b_{j,\ell}=\lambda$ for every cycle $\gamma_j$. Moreover, the action of $\tM$ on $W_j$ can be transformed by a change of basis into a standard weighted cycle. Define
	\[
	\mathbf{f}_{j,1}:=\be_{i_{j,1}}, 
	\mathbf{f}_{j,v}:=(b_{j,1}b_{j,2}\cdots b_{j,v-1})\be_{i_{j,v}} \mbox{ for } 2 \leq v \leq \ell .
	\]
	Then we obtain
	\begin{eqnarray*}
\mathbf{f}_{j,1}\tM&=&\be_{i_{j,1}}\tM=b_{j,1}\be_{i_{j,2}}=\mathbf{f}_{j,2},\\
			\mathbf{f}_{j,v}\tM&=&(b_{j,1}b_{j,2}\cdots b_{j,v-1})\be_{i_{j,v}}\tM=\mathbf{f}_{j,v+1} \mbox{ for } 2\leq v\leq \ell-1,\\
			\mathbf{f}_{j,\ell}\tM&=&(b_{j,1}b_{j,2}\cdots b_{j,\ell-1})\be_{i_{j,\ell}}\tM = \lambda \mathbf{f}_{j,1}.
	\end{eqnarray*}
	Let $F$ denote the set of fixed points of $\tP$. Then $\gf_q^n = \left(\bigoplus_{j=1}^t W_j\right)\oplus \langle\be_i : i \in F \rangle_{\gf_q}$, and all direct summands are $\tM$-invariant. We call $\{\mathbf{f}_{j,1}, \cdots, \mathbf{f}_{j,\ell}\}$ a normalized cyclic basis of $W_j$. With respect to this basis, the restriction $\tM |_{W_j}$ has minimal polynomial $X^{\ell}-\lambda$.   
	Since the minimal polynomial of a restriction divides that of the original transformation, $X^\ell-\lambda$ divides the minimal polynomial of $\tM$. On the other hand, since $\ord(\overline{\tM})=\ell$, we have $\tM^\ell=\lambda\tI_n$, and hence the minimal polynomial of $\tM$ divides $X^\ell-\lambda$. Therefore, the minimal polynomial of $\tM$ is $X^{\ell}-\lambda$.
	
	In the following, we first consider $k$-dimensional linear subspaces $\cC \subseteq \gf_q^n$ of which the automorphism groups contain an element whose image in $H_{\cC}$ has order $2$. For every fixed $\varepsilon>0$, we show that the proportion of such subspaces tends to zero whenever $\min\{k,n-k\}\geq(2+\varepsilon)\log_q n$.
	
	\begin{theorem}\label{thm-ordertwo}
		Fix a prime power $q$ and a real number $\varepsilon>0$. Let $\cC \subseteq \gf_q^n$ be a uniformly random $k$-dimensional linear subspace, where $k=k(n)$, let $m:=\min\{k, n-k\}$, and let $H:=M_{n,q}/\MAut_0$. If $m\ge (2+\varepsilon)\log_q n$, then 
		\[
		\sum_{\substack{
				1\neq \overline \tM\in H\\
				\operatorname{ord}(\overline \tM)=2\\
				\tP\neq \tI_n
		}}
		\Pr\left[\cC\tM=\cC\right]
		\] tends to zero as $n$ tends to infinity.
	\end{theorem}
	
	\begin{IEEEproof}
		Again without loss of generality, we restrict attention to $k \leq n/2$ and $m = k$.
		Let $1\neq \overline \tM\in H$ with $\operatorname{ord}(\overline \tM)=2$. 
		Then there exists $\lambda \in \gf_q^*$ such that $\tM^{2}=\lambda \tI_n$.
		
		\noindent\textbf{Case 1: $\operatorname{char}(\mathbb F_q)=2$.}
		
		By the Frobenius map $x\mapsto x^2$, there exists $\mu\in \mathbb F_q^*$ such that $\mu^2=\lambda^{-1}$. Thus we may use $\mu\tM$ as a representative. For convenience, we still denote it by $\tM=\tD\tP$, and we may assume that $\tM^2=\tI_n$. Suppose that $\tD=\diag(\alpha_1, \dots, \alpha_n)$ with $\alpha_i \in \gf_q^*$, and $\tP$ has $t \geq 1$ transpositions. Without loss of generality, write
		\[
		\tP=(1\ 2)(3\ 4)\cdots(2t-1\ 2t).
		\]
		For a fixed point $i>2t$, we have $\be_i\tM=\alpha_i \be_i$. Since $\tM^2=\tI_n$, it follows that $\alpha_i^2=1$. As $\operatorname{char}(\mathbb F_q)=2$, we have $\alpha_i=1$, and hence $\be_i\tM=\be_i$ for all $i>2t$. Note that for each transposition $(2j-1\ 2j)$, one has $\mathbf{f}_{j,1}\tM=\mathbf{f}_{j,2}$ and $\mathbf{f}_{j,2}\tM=\mathbf{f}_{j,1}$. Let $\tN:=\tM-\tI_n$. Then we have
		\[
		\tN^2=(\tM-\tI_n)^2=\tM^2-2\tM+\tI_n=\tI_n+\tI_n=0.
		\]
		On each transposition $(2j-1\ 2j)$, we have
		$\mathbf{f}_{j,1}\tN=\mathbf{f}_{j,2}-\mathbf{f}_{j,1}$ and
		$\mathbf{f}_{j,2}\tN=\mathbf{f}_{j,1}-\mathbf{f}_{j,2}$.
		Thus the image of $\tN$ on $\langle \mathbf{f}_{j,1}, \mathbf{f}_{j,2} \rangle$ has dimension $1$. 
		It follows that $\dim(\im(\tN))=t$ and $\dim(\Ker(\tN))=n-t$. Since $\tM=\tN+\tI_n$, $\cC\tM=\cC$ if and only if $\cC\tN \subseteq \cC$. By Corollary~\ref{cor-ordertwo}, we have
		\begin{eqnarray*}
			\Pr\left[ \cC\tM=\cC\right] \leq \frac{
				\sum_{r=0}^{\min\{t,\lfloor k/2\rfloor\}}
				\genfrac{[}{]}{0pt}{}{t}{r}_q
				\genfrac{[}{]}{0pt}{}{n-t-r}{k-2r}_q
				q^{(n-t-k+r)r}
			}{
				\genfrac{[}{]}{0pt}{}{n}{k}_q
			}.
		\end{eqnarray*}
		By Lemma~\ref{lem-Gau}, there exist constants $c_1$, $c_2$ such that
		\begin{eqnarray*}
			\genfrac{[}{]}{0pt}{}{t}{r}_q\le c_2 q^{r(t-r)}, \mbox{ }
			\genfrac{[}{]}{0pt}{}{n-t-r}{k-2r}_q
			\le
			c_2 q^{(k-2r)(n-t-r-k+2r)},\mbox{ and }
			\genfrac{[}{]}{0pt}{}{n}{k}_q\ge c_1 q^{k(n-k)}.
		\end{eqnarray*}
		Therefore, we obtain
		\begin{eqnarray*}
			\Pr\left[ \cC\tM=\cC\right]
			\le
			\frac{c_2^2}{c_1}
			\sum_{r=0}^{\min\{t,\lfloor k/2\rfloor\}}
			q^{F(r,t)},
		\end{eqnarray*}
		where
		\begin{eqnarray*}
			F(r,t)
			&:=&
			r(t-r)
			+
			(k-2r)(n-t-r-k+2r)
			+
			(n-t-k+r)r
			-
			k(n-k)\\
			&=&
			t(2r-k)+r(2k-n)-2r^2.
		\end{eqnarray*}
		
		For each nontrivial $\overline{\tM} \in H$ of order $2$ with $\tP \neq \tI_n$, let $\tM=\tD\tP$ be the normalized representative chosen above, satisfying $\tM^2=\tI_n$. We now sum over all such $\tM=\tD\tP$. For fixed $t$, the number of possible permutations $\tP$ with $t$ transpositions and $n-2t$ fixed points is at most
		\[
		\binom{n}{2}\binom{n-2}{2}\cdots \binom{n-2t+2}{2} \leq n^{2t}.
		\]
		Since the weights on each transposition satisfy $\alpha_{2j-1}\alpha_{2j}=1$, the entry $\alpha_{2j-1}$ is uniquely determined by $\alpha_{2j}$. Hence, there are at most $(q-1)^t$ choices for $\tD$. Let $\theta:=\log_q n$ and $\eta:=\log_q(q-1)$. Thus we have
		\begin{eqnarray*}
			\sum_{\substack{
					1\neq \overline \tM\in H\\
					\operatorname{ord}(\overline \tM)=2\\
					\tP\neq \tI_n
			}}
			\Pr\left[ \cC\tM=\cC\right] \leq \frac{c_2^2}{c_1}
			\sum_{t=1}^{\lfloor n/2\rfloor}
			\sum_{r=0}^{\min\{t,\lfloor k/2\rfloor\}}
			q^{G(r,t)},
		\end{eqnarray*}
		where
		$G(r,t)
		:=
		t(2\theta+\eta+2r-k)+r(2k-n)-2r^2.$
		For fixed $r$, the function $G(r,t)$ is linear in $t$. Let
		$
		A_r:=2\theta+\eta+2r-k.
		$
		
			For sufficiently large $n$, $A_0=2\theta+\eta-k\le-\varepsilon\theta/2$. Hence, we have
			\begin{eqnarray}\label{eq-A0}
			    \sum_{t=1}^{\lfloor n/2\rfloor}q^{G(0,t)}
			\le\sum_{t\ge1}q^{-\varepsilon\theta t/2}
			=\frac{1}{n^{\varepsilon/2}-1}=o(1).
			\end{eqnarray}
			Now let $1\le r\le\lfloor k/2\rfloor$. If $A_r\le0$, then $G(r,t)$ is nonincreasing in $t$. Since every admissible $t$ satisfies $r\le t\le\lfloor n/2\rfloor$, we have
			\[
			G(r,t)\le G(r,r)=r(2\theta+\eta+k-n)\le-\frac{rn}{3}
			\]
			for all sufficiently large $n$. Therefore, we have
			\begin{eqnarray}\label{eq-Ar-le0}
			    \sum_{r=1}^{\lfloor k/2\rfloor}
			\sum_{\substack{r\le t\le\lfloor n/2\rfloor\\ A_r\le0}}
			q^{G(r,t)}
			\le\sum_{r\ge1}nq^{-rn/3}=o(1).
			\end{eqnarray}
			If $A_r>0$, then $G(r,t)$ is increasing in $t$, and hence
			\[
			G(r,t)\le G(r,n/2)
			=n\theta+\frac{\eta n}{2}-\frac{nk}{2}+2kr-2r^2.
			\]
			For $0\le r\le k/2$, we have
			\[
			G(r,n/2)
			\le n\theta+\frac{\eta n}{2}-\frac{k(n-k)}{2}.
			\]
			The assumptions $k\le n/2$ and $k\ge(2+\varepsilon)\theta$ imply, for sufficiently large $n$, that
			\[
			k(n-k)\ge\left(2+\frac{\varepsilon}{2}\right)n\theta
			\mbox{ and }
			\eta\le\frac{\varepsilon\theta}{4}.
			\]
			Consequently, $G(r,t)\le-\varepsilon n\theta/8$, and
			\begin{eqnarray}\label{eq-Ar-ge0}
			    \sum_{r=1}^{\lfloor k/2\rfloor}
			\sum_{\substack{r\le t\le\lfloor n/2\rfloor\\ A_r>0}}
			q^{G(r,t)}
			\le n^2q^{-\varepsilon n\theta/8}=o(1).
			\end{eqnarray}
			Combining the inequalities (\ref{eq-A0}), (\ref{eq-Ar-le0}), and (\ref{eq-Ar-ge0}) yields the following estimate.
			\begin{eqnarray*}
				\lefteqn{\sum_{\substack{
							1\neq \overline \tM\in H\\
							\operatorname{ord}(\overline \tM)=2\\
							\tP\neq \tI_n
					}}
					\Pr[\cC\tM=\cC]}\\
				&\le&\frac{c_2^2}{c_1}
				\Bigg(
				\sum_{t=1}^{\lfloor n/2\rfloor}q^{G(0,t)}
				+\sum_{r=1}^{\lfloor k/2\rfloor}
				\sum_{\substack{r\le t\le\lfloor n/2\rfloor\\ A_r\le0}}q^{G(r,t)}
				+\sum_{r=1}^{\lfloor k/2\rfloor}
				\sum_{\substack{r\le t\le\lfloor n/2\rfloor\\ A_r>0}}q^{G(r,t)}
				\Bigg)=o(1).
			\end{eqnarray*}

		\noindent\textbf{Case 2: $\operatorname{char}(\mathbb F_q) \neq 2$.}
		
		In this case, we discuss two subcases: (1) $\lambda$ is a square; (2) $\lambda$ is a nonsquare.
		
		\emph{Subcase 2.1: $\lambda$ is a square.}
		Then there exists $\mu \in \gf_q^*$ such that $\mu^2=\lambda^{-1}$, and we may use $\mu\tM$ as a representative. By abuse of notation, we still denote it by $\tM$, and we may assume that $\tM^2=\tI_n$. Note that $\tM \neq \pm \tI_n$. Then the minimal polynomial of $\tM$ is $X^2-1$. Since $\operatorname{char}(\mathbb F_q) \neq 2$, the polynomial $X^2-1=(X-1)(X+1)$ has no repeated roots. Hence, by Lemma~\ref{lem-PD}, we have
		\[
		\gf_q^n=\Ker(\tM-\tI_n) \oplus \Ker(\tM+\tI_n).
		\]
		Let
		$
		V_{+}:= \Ker(\tM-\tI_n)$ and
		$V_{-}:= \Ker(\tM + \tI_n)
		$.
		For each transposition $(2j-1 \ 2j)$, let $W_j:=\langle \mathbf{f}_{j,1}, \mathbf{f}_{j,2}\rangle_{\mathbb F_q}$. Since
		\[
		(\mathbf{f}_{j,1}+\mathbf{f}_{j,2})\tM=\mathbf{f}_{j,2}+\mathbf{f}_{j,1}=\mathbf{f}_{j,1}+\mathbf{f}_{j,2},
		\]
		and
		\[
		(\mathbf{f}_{j,1}-\mathbf{f}_{j,2})\tM=\mathbf{f}_{j,2}-\mathbf{f}_{j,1}=-(\mathbf{f}_{j,1}-\mathbf{f}_{j,2}),
		\]
		we have $\mathbf{f}_{j,1}+\mathbf{f}_{j,2} \in V_+$ and $\mathbf{f}_{j,1}-\mathbf{f}_{j,2} \in V_-$. Moreover, since $\mathbf{f}_{j,1}$ and $\mathbf{f}_{j,2}$ are linearly independent, we have
		$W_j=\langle \mathbf{f}_{j,1}+\mathbf{f}_{j,2}\rangle \oplus \langle \mathbf{f}_{j,1}-\mathbf{f}_{j,2}\rangle$
		with $\dim(W_j \cap V_+)=1$ and $\dim(W_j \cap V_-)=1$.
		
		Let $s:=\min\{\dim V_+,\dim V_-\}$. The above analysis gives $1 \le t \le s\le \lfloor n/2\rfloor$. If $\cC\tM=\cC$, then we have
		$
		\cC = (\cC \cap V_+) \oplus (\cC \cap V_-).
		$
		Without loss of generality, suppose that $\dim V_-=s$ and $\dim V_+=n-s$. Then the number of $\cC$ satisfying $\cC\tM=\cC$ is at most
		\[
		\sum_{r=0}^{\min\{s,k\}}
		\genfrac{[}{]}{0pt}{}{s}{r}_q
		\genfrac{[}{]}{0pt}{}{n-s}{k-r}_q.
		\]
		It follows that
		\[
		\Pr\left[ \cC\tM=\cC\right]
		\le
		\frac{
			\sum_{r=0}^{\min\{s,k\}}
			\genfrac{[}{]}{0pt}{}{s}{r}_q
			\genfrac{[}{]}{0pt}{}{n-s}{k-r}_q
		}{
			\genfrac{[}{]}{0pt}{}{n}{k}_q
		}.
		\]
		By Lemma~\ref{lem-Gau}, there exist constants $c_1$ and $c_2$ such that
		\[
		\genfrac{[}{]}{0pt}{}{s}{r}_q\le c_2 q^{r(s-r)}, \mbox{ }
		\genfrac{[}{]}{0pt}{}{n-s}{k-r}_q
		\le c_2 q^{(k-r)(n-s-k+r)}, \mbox{ and }
		\genfrac{[}{]}{0pt}{}{n}{k}_q\ge c_1 q^{k(n-k)}.
		\]
		Therefore, we obtain
		\[
		\Pr\left[ \cC\tM=\cC\right]
		\le
		\sum_{r=0}^{\min\{s,k\}} \frac{c_2^2}{c_1}
		q^{F(r,s)},
		\]
		where
		$
		F(r,s) := s(2r-k)+r(2k-n)-2r^2.
		$

        We now sum over all nontrivial cosets $\overline{\tM} \in H$ of order $2$ with $\tP \neq \tI_n$ for which $\lambda_{\tM}$ is a square, using for each coset the normalized representative $\tM=\tD\tP$ chosen above, i.e., $\tM^2=\tI_n$. Suppose that $\tP$ has $t$ transpositions and $n-2t$ fixed points, and that the standard basis vectors corresponding to $b$ of these fixed points lie in $V_-$. Then we have $s=t+b$ and $2t+b \leq 2s$. By counting techniques, the number of permutations $\tP$ with $t$ transpositions and $b$ fixed points whose corresponding standard basis vectors lie in $V_-$ is at most $n^{2s}$, and there are at most $(q-1)^t$ choices for $\tD$. In addition, for fixed $s$, there are at most $s$ choices for $t$. Let $\theta:=\log_q n$ and $\eta:=\log_q(q-1)$. Denote by $
		G(r,s):= s(2\theta+\eta+2r-k)+r(2k-n)-2r^2.
		$ Since $s \leq n/2$, for $r > k/2$ we have $G(r,s) - G(k-r,s)=(2r-k)(2s-n) \leq 0$. 
        Therefore, if $S_{\mathrm{sq}}$ denotes the contribution of the cosets for which $\lambda_{\tM}$ is a square, then
		\begin{eqnarray*}
			S_{\mathrm{sq}} \leq \frac{c_2^2}{c_1}
			\sum_{s= 1}^{\lfloor n/2\rfloor} \sum_{r=0}^{\min\{s,k\}} s\,q^{G(r,s)}
			\leq \frac{2c_2^2}{c_1}
			\sum_{s= 1}^{\lfloor n/2\rfloor} \sum_{r=0}^{\min\{s,\lfloor k/2\rfloor\}} s\,q^{G(r,s)}.
		\end{eqnarray*} By the arguments similar to those used to derive (\ref{eq-A0}), (\ref{eq-Ar-le0}), and (\ref{eq-Ar-ge0}), we obtain that $S_{\mathrm{sq}}$ tends to $0$ as $n$ tends to infinity.
		
		\emph{Subcase 2.2: $\lambda$ is a nonsquare.}
		If $i$ is a fixed point of $\tP$, then $\alpha_i^2=\lambda$, contradicting the assumption that $\lambda$ is a nonsquare. Hence, the permutation $\tP$ must consist of $n/2$ transpositions. Without loss of generality, write
		\[
		\tP =(1\ 2)(3\ 4)\cdots (n-1\ n).
		\]
		For each transposition $(2j-1 \ 2j)$, let
		$W_j:=\langle \mathbf{f}_{j,1},\mathbf{f}_{j,2}\rangle_{\mathbb F_q}$ for $1\le j\le n/2.
		$
		Then $\tM$ acts on $W_j$ by
		\[
		\mathbf{f}_{j,1}\mapsto \mathbf{f}_{j,2},\
		\mathbf{f}_{j,2}\mapsto \lambda \mathbf{f}_{j,1}.
		\]
		Since $\lambda$ is not a square, the quadratic polynomial $X^2-\lambda$ is irreducible in $\mathbb F_q[X]$. Therefore, we have
		\[
		K=\mathbb F_q[X]/(X^2-\lambda) = \gf_q(\omega) = \langle 1, \omega \rangle_{\mathbb F_q}
		\]
		is a quadratic extension of $\gf_q$, where $\omega^2=\lambda$. Define $\psi_j: W_j \rightarrow K$ by $\psi_j(\mathbf{f}_{j,1})=1$ and $\psi_j(\mathbf{f}_{j,2})=\omega$. Since $\{\mathbf{f}_{j,1}, \mathbf{f}_{j,2}\}$ are $\gf_q$-bases of $W_j$ and $K$, respectively, $\psi_j$ is an $\gf_q$-isomorphism, and hence $W_j \cong K$. Under this isomorphism, the action of $\tM$ on $W_j$ is scalar multiplication in $K$, namely multiplication by $\omega$. Moreover, we have
		\[
		\gf_q^n=W_1\oplus\cdots\oplus W_{n/2} \cong K^{n/2}.
		\]
		Thus, $\gf_q^n$ may be regarded as an $n/2$-dimensional linear space over $K$. A $k$-dimensional $\gf_q$-linear subspace $\cC$ satisfying $\cC\tM=\cC$ is then a $k/2$-dimensional $K$-linear subspace. Hence, the number of such a linear code $\cC$ is at most
		\[
		\genfrac{[}{]}{0pt}{}{n/2}{k/2}_{q^2}.
		\]
		It follows that
		\begin{eqnarray*}
			\Pr\left[ \cC\tM=\cC\right]
			\le 
			\frac{
				\genfrac{[}{]}{0pt}{}{n/2}{k/2}_{q^2}
			}{
				\genfrac{[}{]}{0pt}{}{n}{k}_q
			}.
		\end{eqnarray*}
		By Lemma~\ref{lem-Gau}, there exist constants $c_1$, $c_2$ such that
		\begin{eqnarray*}
			\genfrac{[}{]}{0pt}{}{n/2}{k/2}_{q^2}
			\le
			c_2
			(q^2)^{\frac k2(\frac n2-\frac k2)}
			=
			c_2q^{\frac{k(n-k)}{2}} \mbox{ and }
			\genfrac{[}{]}{0pt}{}{n}{k}_q
			\ge
			c_1 q^{k(n-k)}.
		\end{eqnarray*}
		Therefore, we obtain
		\begin{eqnarray*}
			\Pr\left[ \cC\tM=\cC\right]
			\le 
			\frac{c_2}{c_1}
			q^{-\frac{k(n-k)}{2}}.
		\end{eqnarray*}
		
        We next estimate the total contribution of the nontrivial cosets $\overline{\tM} \in H$ of order $2$ with $\tP \neq \tI_n$ for which $\lambda_{\tM}$ is a nonsquare, choosing one representative $\tM=\tD\tP$ from each coset. Since there are at most $n^n$ choices for $\tP$ with $n/2$ transpositions, and there are at most $(q-1)^{n/2}$ choices for $\tD$, the number of such $\tM$ is at most $(q-1)^{n/2}n^n$. Let $\theta:=\log_q n$ and $\eta:=\log_q(q-1)$. If $S_{\mathrm{nsq}}$ denotes the contribution of the cosets for which $\lambda_{\tM}$ is a nonsquare, then
		\begin{eqnarray*}
			S_{\mathrm{nsq}}
			\leq \frac{c_2}{c_1}
			q^{\frac{n \eta +2n\theta-k(n-k)}{2}}.
		\end{eqnarray*}
		Since $k\ge (2+\varepsilon)\theta$, it follows that
		$
		k(n-k)
		\ge
		(2+\varepsilon)\theta(n-(2+\varepsilon)\theta)
		=
		(2+\varepsilon)n\theta-(2+\varepsilon)^2\theta^2.
		$
		For sufficiently large $n$, we have
		$
		(2+\varepsilon)^2\theta^2\le \varepsilon \theta n/4.
		$
		Thus
		$
		k(n-k)
		\ge
		\left(2+\frac{3\varepsilon}{4}\right)n\theta.
		$
		On the other hand, we have $\eta \leq \varepsilon \theta/2$ for sufficiently large $n$. It then follows that
		\[
		\frac{n \eta +2n\theta-k(n-k)}{2}
		\leq
		\frac{1}{2} \left(\frac{\varepsilon n \theta}{2} + 2n\theta -\left(2+\frac{3\varepsilon}{4}\right)n\theta \right)
		= -\frac{\varepsilon n\theta}{8}.
		\]
		Consequently, we obtain
		\begin{eqnarray*}
			S_{\mathrm{nsq}} \leq \frac{c_2}{c_1}
			q^{\frac{n \eta +2n\theta-k(n-k)}{2}} \le \frac{c_2}{c_1} \cdot
			q^{-\frac{\varepsilon n\theta}{8}},
		\end{eqnarray*} 
        and we thereby have $S_{\mathrm{nsq}}$ tends to 0 as $n$ goes to infinity. Therefore, we have in odd characteristic, 
		\[
		\sum_{\substack{
				1\neq \overline \tM\in H\\
				\operatorname{ord}(\overline \tM)=2\\
				\tP\neq \tI_n
		}}
		\Pr[\cC\tM=\cC]
		=S_{\mathrm{sq}}+S_{\mathrm{nsq}} \to 0,
		\]
        as $n$ goes to infinity.
	\end{IEEEproof}
	
	
	We now consider $k$-dimensional linear subspaces $\cC \subseteq \gf_q^n$ that have automorphism groups containing an element whose image in $H_{\cC}$ has prime order $\ell \geq 3$, and show that the proportion of such subspaces also goes to zero as $n$ goes to infinity, provided that $\min\{k, n-k\} \geq (2+\varepsilon)\log_q n$ for any fixed $\varepsilon>0$.

	\begin{theorem}\label{thm-primeorder}
	Fix a prime power $q$ and a real number $\varepsilon>0$. Let $\cC \subseteq \gf_q^n$ be a uniformly random $k$-dimensional linear subspace, where $k=k(n)$, let $m:=\min\{k, n-k\}$, and let $H:=M_{n,q}/\MAut_0$.  Define $\mathcal P:=\{\ell:\ell\ \text{prime},\ 3 \leq \ell \leq n\}.$  If $m\ge (2+\varepsilon)\log_q n$, 
		then
		\[
		\sum_{\ell \in \mathcal{P} }\sum_{\substack{
				1\neq \overline \tM\in H\\
				\operatorname{ord}(\overline \tM) = \ell\\
				\tP\neq \tI_n
		}}
		\Pr\left[ \cC\tM=\cC\right]
		\] tends to zero as $n$ tends to infinity.
	\end{theorem}
	
	\begin{IEEEproof}
        It suffices to consider $k \leq n/2$, in which case $m=k$. 
		Let $1\neq \overline \tM\in H$ with $\operatorname{ord}(\overline \tM)=\ell \geq 3$, where $\ell$ is prime. Choose a representative $\tM=\tD\tP\in M_{n,q}$, where $\tD=\operatorname{diag}(\alpha_1,\dots,\alpha_n)$ with $\alpha_i\in \mathbb F_q^*$ and $\tP \in S_n$. Since $\operatorname{ord}(\overline \tM)=\ell$, there exists $\lambda \in \gf_q^*$ such that
		$\tM^{\ell}=(\tD\tP)^{\ell}=\lambda\tI_n$.
		Thus, the permutation $\tP$ can be written as $\tP=\gamma_1\gamma_2\cdots\gamma_t$, where $t\geq 1$, the $\gamma_j$'s are pairwise disjoint cycles of length $\ell$, and the remaining coordinates are fixed by $\tP$. Let
		$\gamma_j=(i_{j,1}, i_{j,2}, \cdots, i_{j, \ell})$ and
		$W_j=\langle \mathbf{f}_{j,1}, \mathbf{f}_{j,2}, \cdots, \mathbf{f}_{j,\ell}\rangle_{\gf_q}$.
		The minimal polynomial of $\tM |_{W_j}$ is $X^{\ell}-\lambda$. Moreover, the minimal polynomial of $\tM$ is also $X^{\ell}-\lambda$. 
		
		\medskip
		\noindent\textbf{Case 1: $\ell \neq \operatorname{char}(\mathbb F_q)$.}
		
		\emph{Subcase 1.1: $X^\ell-\lambda$ has a root in $\mathbb F_q$.}
		
		Let 
		$
		F:=\{i\in[n]: i\tP=i\}$ be the set of fixed points. Then we have $|F|=n-\ell t.
		$
		For every $i \in F$, there exists $\alpha_i \in \gf_q^*$ such that $\be_i\tM=\alpha_i \be_i$, and hence $\be_i\tM^\ell=\lambda \be_i$. On the other hand, $\tM^\ell=\lambda \tI_n$ gives $\be_i\tM^\ell=\alpha_i^\ell \be_i$. Therefore, we have $\alpha_i^\ell=\lambda$, so the scalar on each fixed point is a root of $X^\ell-\lambda$ in $\mathbb F_q$. 
		We partition $F$ into subsets according to the value of $\alpha_i$:
		\[
		F = \bigcup\limits_{\alpha \in \gf_q} F_{\alpha}, \
		F_\alpha=\{i\in F:\alpha_i=\alpha\}.
		\]
		Choose a root $\alpha_0$ such that
		$
		b_0=|F_{\alpha_0}|=\max |F_\alpha|.
		$
		Define
		$
		\bar{f}:=|F|-b_0=n-\ell t-b_0.
		$
		Let $U:=\Ker(\tM-\alpha_0 \tI_n)$. Since $\alpha_0$ is an eigenvalue of $\tM |_{W_j}$ and $X^{\ell}-\lambda$ has no repeated roots, we have $\dim(W_j\cap U)=1$. For a fixed point $i\in F$, $\be_i\in U$ if and only if $\alpha_i=\alpha_0$. It follows that $\dim U=t+b_0$. Define $\bar{n}:=n-\dim U$. Then we have
		$\bar{n} = n-(t+b_0)=\ell t+|F|-t-b_0=(\ell-1)t+\bar{f}$.
		
		Let
		$
		h(X)=\frac{X^\ell-\lambda}{X-\alpha_0}.
		$
		Since $X^\ell-\lambda$ has no repeated roots, we obtain $\gcd(X-\alpha_0,h(X))=1$. 
		Then we have the direct sum decomposition
		\[
		\mathbb F_q^n
		=
		\Ker(\tM-\alpha_0 \tI_n)
		\oplus
		\Ker (h(\tM))
		\] induced by the Bezout's identity.
		Let $W:=\Ker (h(\tM))$. Then $\mathbb F_q^n=U\oplus W$ with $\dim U=n-\bar{n}$ and $\dim W=\bar{n}$. Clearly, both $U$ and $W$ are $\tM$-invariant subspaces. If $\mathcal C \tM=\mathcal C$, then we have
		\[
		\mathcal C=(\mathcal C\cap U)\oplus(\mathcal C\cap W).
		\]
		
		We now count the $k$-dimensional subspaces $\cC$ satisfying $\mathcal C \tM=\mathcal C$. Let $r:=\dim(\mathcal C\cap W)$. Then $\dim(\mathcal C\cap U)=k-r$. Since $\dim U=n-\bar{n}$, we have
		$\max\{0,k-(n-\bar{n})\}\le r\le \min\{k,\bar{n}\}$.
		For fixed $r$, there are $\genfrac{[}{]}{0pt}{}{\bar{n}}{r}_q$ ways to choose $\mathcal C\cap W$ and $\genfrac{[}{]}{0pt}{}{n-\bar{n}}{k-r}_q$ ways to choose $\mathcal C\cap U$, respectively. Hence, the number of $k$-dimensional subspaces satisfying $\mathcal C \tM=\mathcal C$ is at most
		\[
		\sum_{r}\genfrac{[}{]}{0pt}{}{\bar{n}}{r}_q\genfrac{[}{]}{0pt}{}{n-\bar{n}}{k-r}_q,
		\]
		where $r$ ranges over the legal interval above. For convenience, write this as
		\[
		\sum_{r=0}^{\min\{k,\bar{n}\}}
		\genfrac{[}{]}{0pt}{}{\bar{n}}{r}_q
		\genfrac{[}{]}{0pt}{}{n-\bar{n}}{k-r}_q,
		\]
		with the convention that the Gaussian coefficient is $0$ when $k-r\notin[0,n-\bar{n}]$. It follows that
		\[
		\Pr\left[ \cC\tM=\cC\right]
		\le
		\frac{
			\displaystyle
			\sum_{r=0}^{\min\{k,\bar{n}\}}
			\genfrac{[}{]}{0pt}{}{\bar{n}}{r}_q
			\genfrac{[}{]}{0pt}{}{n-\bar{n}}{k-r}_q
		}{
			\genfrac{[}{]}{0pt}{}{n}{k}_q
		}.
		\]
		By Lemma~\ref{lem-Gau}, there exist constants $c_1$, $c_2$ such that
		\[
		\genfrac{[}{]}{0pt}{}{n}{k}_q
		\ge
		c_1 q^{k(n-k)}, \
		\genfrac{[}{]}{0pt}{}{\bar{n}}{r}_q
		\le
		c_2 q^{r(\bar{n}-r)}, \mbox{ and }
		\genfrac{[}{]}{0pt}{}{n-\bar{n}}{k-r}_q
		\le
		c_2 q^{(k-r)(n-\bar{n}-k+r)}.
		\]
		Thus we obtain
		\begin{eqnarray}\label{eq-twopd}
		    \Pr\left[ \cC\tM=\cC\right]
		\le
		\frac{c_2^2}{c_1}
		\sum_{r=0}^{\min\{k,\bar{n}\}}
		q^{-\Delta(r,\bar{n})},
		\end{eqnarray}
		where
		\begin{eqnarray*}
			\Delta(r,\bar{n}) &:=& k(n-k)
			-
			r(\bar{n}-r)
			-
			(k-r)(n-\bar{n}-k+r) \\
			& = & k\bar{n}+r(n-2k-2\bar{n})+2r^2 \\
            & = & (k-r)(\bar{n}-r)+r(n-\bar{n}-k+r).
		\end{eqnarray*}
		It is clear that $\Delta(r,\bar{n}) \geq 0$. 
		
		If $k+\bar{n}\le n/2$, then for every $r$,
		$
		\Delta(r,\bar{n})=k\bar{n}+r(n-2k-2\bar{n})+2r^2 \geq k\bar{n}.
		$
		It follows that
		\[
		\Pr\left[ \cC\tM=\cC\right]
		\le
		\frac{c_2^2}{c_1} \sum_{r=0}^{\min\{k,\bar{n}\}}
		q^{-\Delta(r,\bar{n})} \leq \frac{c_2^2}{c_1}
		(k+1)q^{-k\bar{n}}.
		\]
		
		If $k+\bar{n}> n/2$, the preceding estimate (in the inequality (\ref{eq-twopd})) does not yield a sufficiently strong bound in certain subcases. A different argument is therefore needed for these cases. Let $X^\ell-\lambda=f_0(X)f_1(X)\cdots f_s(X)$ be the irreducible factorization of $X^{\ell}-\lambda$ in $\gf_q[X]$, where $f_0(X)=X-\alpha_0$. 
		By Lemma~\ref{lem-PD}, we have $\mathbb F_q^n=V_0\oplus V_1\oplus\cdots\oplus V_s,$ where $V_0=U=\Ker(\tM-\alpha_0 \tI_n),$ and $V_i=\Ker(f_i(\tM))$ for $1 \le i \le s$. Note that $\dim V_0 =n-\bar{n}$.
		Let $d_i:=\deg f_i$. Since $V_i=\Ker(f_i(\tM))$, the action of $\tM$ endows $V_i$ with an $\mathbb F_q[X]/(f_i)$-module structure. Indeed, for $\bv \in V_i$ and $g(X) \in \mathbb F_q[X]$, the action is given by $\bv \cdot \overline{g(X)}:= \bv g(\tM)$, where $\overline{g(X)} \in \mathbb F_q[X]/(f_i)$. Since $f_i(X)$ is irreducible of degree $d_i$, we have $\mathbb F_q[X]/(f_i) \cong \mathbb F_{q^{d_i}}$. Then $V_i$ may be regarded as an $\mathbb F_{q^{d_i}}$-vector space. Denote by $a_i:=\dim_{\mathbb F_{q^{d_i}}}V_i.$ 
		If $d_i = 1$, then there exists $\alpha_i \in \gf_q^*$ such that $f_i(X)=X-\alpha_i$. Then we have $\dim V_i = t+ |F_{\alpha_i}| \leq t+|F_{\alpha_0}| = n-\bar{n}$. 
		If $d_i >1$, the polynomial $f_i$ has no roots in $\gf_q$. It then follows that $V_i \cap \langle\be_i : i \in F \rangle_{\gf_q} = \{\bzero\}$. 
		Since $\gf_q^n = \left(\bigoplus_{j=1}^t W_j\right)\oplus \langle\be_i : i \in F \rangle_{\gf_q}$ and all summands are $\tM$-invariant, they are also $f_i(\tM)$-invariant. Hence, we have $V_i=\bigoplus_{j=1}^t (W_j \cap V_i)$.
		Besides, for each $W_j$, we have
		$$W_j\cong \mathbb F_q[X]/(X^\ell-\lambda) 
		\cong
		\bigoplus_{i=0}^s \mathbb F_q[X]/(f_i)$$ as an $\mathbb F_q[X]$-module. Therefore, we obtain $W_j \cap V_i \cong \mathbb F_q[X]/(f_i)\cong \mathbb F_{q^{d_i}}$. 
		It follows that $\dim_{\mathbb F_{q^{d_i}}}(W_j \cap V_i)=1$ and $$a_i=\dim_{\mathbb F_{q^{d_i}}}V_i=\sum_{j=1}^{t}\dim_{\mathbb F_{q^{d_i}}}(W_j \cap V_i)=t \le t+|F_{\alpha_0}| = n-\bar{n}.$$ Hence, $a_i \leq n-\bar{n}$ for all $0 \leq i \leq s$.
		Now we count the $k$-dimensional subspaces $\cC$ satisfying $\cC\tM=\cC$. If $\cC\tM=\cC$, then $\mathcal C
		= 
		\bigoplus_{i=0}^{s}(\mathcal C\cap V_i).$ Moreover, each $\cC \cap V_i$ is an $\gf_{q^{d_i}}$-subspace of $V_i$.
		Let $r_i:=\dim_{\mathbb F_{q^{d_i}}}(\mathcal C\cap V_i)$. Then $\sum_{i=0}^{s} d_i r_i=k.$ For fixed $(r_0,\dots,r_s)$, the number of choices for $\mathcal C\cap V_i$ is $\left[\begin{matrix} a_i \\ r_i \end{matrix}\right]_{q^{d_i}}$, and then the number of $k$-dimensional subspaces satisfying $\mathcal C \tM=\mathcal C$ is at most
		$$\sum_{\sum d_i r_i=k}
		\prod_{i=0}^{s}
		\left[\begin{matrix} a_i \\ r_i \end{matrix}\right]_{q^{d_i}}.$$
		By Lemma \ref{lem-Gau}, there exists a constant $c_2$ such that
		\begin{eqnarray*}
			\left[\begin{matrix} a_i \\ r_i \end{matrix}\right]_{q^{d_i}}
			\le
			c_2 (q^{d_i})^{r_i(a_i-r_i)}
			=
			c_2 q^{d_i r_i(a_i-r_i)}.
		\end{eqnarray*}
		Then $\prod\limits_{i=0}^{s}
		\left[\begin{matrix} a_i \\ r_i \end{matrix}\right]_{q^{d_i}}
		\le
		c_2^{s+1}
		q^{\sum_i d_i r_i(a_i-r_i)}.$ Since $a_i-r_i\le a_i\le n-\bar{n}$, we obtain $d_i r_i(a_i-r_i)
		\le
		d_i r_i(n-\bar{n})$. It follows that $$\sum_i d_i r_i(a_i-r_i)
		\le
		(n-\bar{n})\sum_i d_i r_i
		=
		k(n-\bar{n}).$$ Moreover, the number of $(r_0,\dots,r_s)$ such that $\sum d_i r_i=k$ is at most $(k+1)^{s+1}\le (k+1)^{\bar{n}+1}$ as $s+1\le \ell\le \bar{n}+1$.
		Hence, $|\{\mathcal C\in \operatorname{Gr}_q(k,n):\mathcal C \tM=\mathcal C\}|
		\le
		c_2^{\bar{n}+1}(k+1)^{\bar{n}+1}
		q^{k(n-\bar{n})}$, and then we have 
		\begin{eqnarray*}
			\Pr\left[ \cC\tM=\cC\right]
			\le
			\frac{c_2^{\bar{n}+1}}{c_1}(k+1)^{\bar{n}+1}
			q^{-k(\bar{n}-k)}.
		\end{eqnarray*} 
		Note that this estimate is strong when $\bar{n} > k$. Similar to the analysis of $\Delta(r,\bar{n})$ and $k(\bar{n}-k)$ in the proof of Theorem \ref{thm-I}, we obtain
		$
		\Pr\left[ \cC\tM=\cC\right] \leq p_{\bar{n}},
		$ where
		\begin{eqnarray}\label{eq-pbarn2}
			p_{\bar{n}}=
			\begin{cases}
				\frac{c_2^2}{c_1}(k+1)q^{-k\bar{n}}, 
				& k+\bar{n}\le n/2,\\[4pt]
				\frac{c_2^{\bar{n}+1}}{c_1}(k+1)^{\bar{n}+1}q^{-k(\bar{n}-k)}, 
				& k+\bar{n}>n/2,\ k\le \theta^2,\\[4pt]
				\frac{c_2^{\bar{n}+1}}{c_1}(k+1)^{\bar{n}+1}q^{-k\bar{n}/4}, 
				& k+\bar{n}>n/2,\ k>\theta^2.
			\end{cases}
		\end{eqnarray}
		
		We now sum the above probability bound over all nontrivial cosets $\overline{\tM} \in H$ of prime order $\ell \geq 3$ for which $X^{\ell} - \lambda$ has a root in $\gf_q$, using one chosen representative $\tM=\tD\tP$ for each coset. For fixed $\bar{n}$, we have $\ell \leq \bar{n}+1$, $t \leq \bar{n}/(\ell-1)$ for each $\ell$, and $\bar{f}=\bar{n}-(\ell-1)t$ is determined by $\bar{n}, \ell, t$, then the number of $(\ell,t,\bar{f})$ is at most $(\bar{n}+1)^2$. For fixed $(\ell,t,\bar{f})$, the number of permutations $\tP$ having $t$ $\ell$-cycles and $\bar{f}$ fixed points outside $F_{\alpha_0}$ is at most $n^{\ell t+\bar{f}}$. Since $\alpha_{i_{j,1}}\alpha_{i_{j,2}}\cdots \alpha_{i_{j,\ell}}=1$ for every $\ell$-cycle $\gamma_j=(i_{j,1}, i_{j,2}, \cdots, i_{j, \ell})$ and the choices of scalar for each fixed points not in $F_{\alpha_0}$ is at most $q-1$, there are at most $(q-1)^{(\ell-1)t+\bar{f}}=(q-1)^{\bar{n}}$ choices for $\tD$.
		Since $\ell t+\bar{f}=(\ell-1)t+\bar{f}+t=\bar{n}+t$ and $t\le  \bar{n}/(\ell-1)\le  \bar{n}/2$, it follows that $\ell t+\bar{f}\le 3\bar{n}/2$. Therefore, for fixed $\bar{n}$, there are at most $(\bar{n}+1)^2 n^{3\bar{n}/2}(q-1)^{\bar{n}}$ choices for $\tM$. Define
		\begin{eqnarray*}
			\mathcal{P}_1:=\{ \ell: \ell \ \text{is prime}, \ 3 \le \ell \le n,\ \ell \neq \Char(\gf_q),\ X^{\ell}-\lambda\ \text{has a root in } \gf_q \}.
		\end{eqnarray*}
		Then we have
		\begin{eqnarray*}
			S_{\mathrm{root}}:=\sum_{\ell \in \mathcal{P}_1} 
			\sum_{\substack{1\ne \tM\in H \\ \operatorname{ord}(\overline{\tM})=\ell\\ \tP\ne \tI_n}}
			\Pr\left[ \cC\tM=\cC\right]
			\le
			\sum_{\bar{n}\ge 2}
			(\bar{n}+1)^2q^{\bar{n}(\frac {3\theta}{2}+\eta)}p_{\bar{n}},
		\end{eqnarray*} where $\theta:=\log_q n$ and $\eta:=\log_q(q-1).$
		To estimate $S_{\mathrm{root}}$, we need separately consider the following two parts:
		\begin{eqnarray*}
			S_1=
			\sum_{\substack{ \bar{n} \ge 2 \\ k+\bar{n}\le n/2}}
			(\bar{n}+1)^2q^{\bar{n}(\frac {3\theta}{2}+\eta)}p_{\bar{n}}, \mbox{ } S_2=
			\sum_{\substack{ \bar{n} \ge 2 \\ k+\bar{n}> n/2}}
			(\bar{n}+1)^2q^{\bar{n}(\frac {3\theta}{2}+\eta)}p_{\bar{n}}.
		\end{eqnarray*}
		
		If $k+\bar{n}\le n/2$, then it follows from Eq. (\ref{eq-pbarn2}) that
		\begin{eqnarray*}
			S_1 \le \frac{c_2^2}{c_1}(k+1)\sum_{\substack{ \bar{n} \ge 2 \\ k+\bar{n}\le n/2}}
			(\bar{n}+1)^2q^{\bar{n}(\frac {3\theta}{2}+\eta-k)}.
		\end{eqnarray*} 
		Since $k\ge (2+\varepsilon)\theta$, and $\eta\le \varepsilon \theta/2$ for sufficiently large $n$, we have
		$$\frac32\theta+\eta-k
		\le
		\frac32\theta+\frac{\varepsilon}{2}\theta-(2+\varepsilon)\theta
		=
		-\left(\frac12+\frac{\varepsilon}{2}\right)\theta.$$
		Hence, we obtain
		\begin{eqnarray*}
			S_1
			\le
			\frac{c_2^2}{c_1}(k+1)
			\sum_{\bar{n}\ge 2}
			(\bar{n}+1)^2q^{-\frac{(1+\varepsilon)\bar{n}\theta}{2}}=\frac{c_2^2}{c_1}(k+1)\sum_{\bar{n}\ge 2}
			(\bar{n}+1)^2 n^{-\frac{(1+\varepsilon)\bar{n}}{2}}.
		\end{eqnarray*}  
		Let $c=(1+\varepsilon)/2$. There exists a constant  $c'$ such that
		\begin{eqnarray*}
			\sum_{\bar{n}\ge 2}(\bar{n}+1)^2n^{-c\bar{n}} = n^{-2c}\sum_{\bar{n}\ge 2}(\bar{n}+1)^2n^{-c(\bar{n}-2)}
			\le
			c' n^{-2c}
			=
			c' n^{-1-\varepsilon}
		\end{eqnarray*}  for sufficiently large $n$. This gives                 
		$$S_1 
		\le
		\frac{c_2^2}{c_1} c' (k+1)n^{-1-\varepsilon} \le \frac{c_2^2}{c_1} c' n^{-\varepsilon}, $$ which tends to zero as $n$ tends to infinity.
		
		If $k+\bar{n}>n/2$ and $k\le \theta^2$, then we have $\bar{n}>\frac n2-k \geq \frac{n}{3}$ for sufficiently large $n$. It then follows from Eq. (\ref{eq-pbarn2}) that
		\begin{eqnarray*}
			S_2
			\le
			\sum_{\bar{n}\ge n/3} \frac{c_2^{\bar{n}+1}}{c_1}(\bar{n}+1)^2(k+1)^{\bar{n}+1}
			q^{\bar{n}(\frac{3\theta}{2}+\eta)-k(\bar{n}-k)}.
		\end{eqnarray*} 
		Since $k\le \theta^2$ and $\bar{n}\ge n/3$, we have $k^2 \leq \varepsilon \bar{n}\theta/4$, $\eta\le \varepsilon \theta/8$, and $\frac{c_2^{\bar{n}+1}}{c_1}(\bar{n}+1)^2(k+1)^{\bar{n}+1}
		\le
		q^{\varepsilon \bar{n}\theta/8}$ for sufficiently large $n$. Besides, $k\ge (2+\varepsilon)\theta$. Then we have $k(\bar{n}-k)=k\bar{n}-k^2 
		\ge
		\left(2+\frac{3\varepsilon}{4}\right)\bar{n}\theta.$  It follows that
		\begin{eqnarray*}
			S_2
			\le
			\sum_{\bar{n}\ge n/3}
			q^{-(4+5\varepsilon)\bar{n}\theta/8}
			q^{\varepsilon \bar{n}\theta/8}=
			\sum_{\bar{n}\ge n/3}
			q^{-(1+\varepsilon)\bar{n}\theta/2} \le 
			nq^{-(1+\varepsilon)n\theta/6} =n^{-\frac{(1+\varepsilon)n}{6} +1},
		\end{eqnarray*} which tends to zero as $n$ tends to infinity.
		
		If $k+\bar{n}>n/2$ and $k > \theta^2$, it follows from Eq. (\ref{eq-pbarn2}) that
		\begin{eqnarray*}
			S_2
			\le
			\sum_{\bar{n}\ge 2}
			\frac{c_2^{\bar{n}+1}}{c_1}(\bar{n}+1)^2
			q^{\bar{n}(\frac{3\theta}{2}+\eta)}
			(k+1)^{\bar{n}+1}
			q^{-k\bar{n}/4}.
		\end{eqnarray*}
		Since $ \theta^2 < k\le n/2$, we have $\log_q(k+1)\le \theta$ and
		\begin{eqnarray*}
			\frac{c_2^{\bar{n}+1}}{c_1}
			(k+1)^{\bar{n}+1}
			(\bar{n}+1)^2 q^{\bar{n}(\frac{3\theta}{2}+\eta)}
			\le
			q^{4\bar{n}\theta} \leq q^{k\bar{n}/8}
		\end{eqnarray*} for sufficiently large $n$. It follows that
		\begin{eqnarray*}
			S_2 \le 
			\sum_{\bar{n}\ge 2}q^{-k\bar{n}/8}
			\le
			\frac{q^{-k/4}}{1-q^{-k/8}},
		\end{eqnarray*} which tends to zero as $n$ tends to infinity.
		
		By the above discussions, we conclude that 
		\begin{eqnarray*}
			S_{\mathrm{root}}:=\sum_{\ell \in \mathcal{P}_1} 
			\sum_{\substack{1\ne \tM\in H \\ \operatorname{ord}(\overline{\tM})=\ell\\ \tP\ne \tI_n}}
			\Pr\left[ \cC\tM=\cC\right] \le S_1+S_2,
		\end{eqnarray*} which also tends to zero as $n$ tends to infinity.
		
		\medskip
		\emph{Subcase 1.2: $X^\ell-\lambda$ has no roots in $\mathbb F_q$.}
		
		Since $X^\ell-\lambda$ has no roots in $\mathbb F_q$, it is easy to deduce that $\tP$ has no fixed points. Then $n=\ell t$ with $t \geq 1$ in this case. Since $\ell\neq \operatorname{char}(\mathbb F_q)$ and $X^\ell-\lambda$ has no roots in $\gf_q$, write $X^\ell-\lambda=f_1(X)\cdots f_s(X)$ for its factorization into irreducible polynomials over $\gf_q$, where $f_1, f_2, \cdots, f_s$ are irreducible and pairwise distinct, and $d_i:=\deg f_i\ge 2$. In particular, $\sum_{i=1}^s d_i=\ell$ and $s\le \ell/2$. Let $V_i:=\Ker f_i(\tM)$ for all $1 \leq i \leq s$. Then we have $\mathbb F_q^n
		=
		V_1\oplus V_2\oplus\cdots\oplus V_s$ by Lemma~\ref{lem-PD}. By an argument similar to that used for the second estimate in Subcase 1.1, the number of $k$-dimensional subspaces satisfying $\mathcal C \tM=\mathcal C$ is at most
		\begin{eqnarray*}
			\sum_{\sum d_ir_i=k}
			\prod_{i=1}^s
			\left[\begin{matrix}
				t\\ r_i
			\end{matrix}\right]_{q^{d_i}},
		\end{eqnarray*} where $r_i:=\dim_{\mathbb F_{q^{d_i}}}(\mathcal C\cap V_i)$.
		By Lemma \ref{lem-Gau}, there exists a constant $c_2$ such that
		\begin{eqnarray*}
			\prod_{i=1}^s
			\left[\begin{matrix}
				t\\ r_i
			\end{matrix}\right]_{q^{d_i}}
			\le
			c_2^s
			q^{\sum d_ir_i(t-r_i)}.
		\end{eqnarray*}
		Let $k_i:=d_ir_i$ and $n_i:=d_it$. Then we have $\sum_{i=1}^s k_i=k,
		\sum_{i=1}^s n_i=n$, and $d_ir_i(t-r_i)=\frac{k_i(n_i-k_i)}{d_i}.$ Since $d_i \geq 2$, we have $d_ir_i(t-r_i)
		\le
		\frac12 k_i(n_i-k_i)$, and then $\sum_{i=1}^s d_ir_i(t-r_i)
		\le
		\frac12
		\sum_{i=1}^s k_i(n_i-k_i).$ Moreover, it is clear that $$\sum_{i=1}^s k_i(n_i-k_i)
		\le
		\left(\sum_{i=1}^s k_i\right)
		\left(\sum_{i=1}^s(n_i-k_i)\right)= k(n-k).$$
		Thus, we obtain $\sum_{i=1}^s d_ir_i(t-r_i)
		\le
		\frac12 k(n-k)$ and 
		\begin{eqnarray*}
			\prod_{i=1}^s
			\left[\begin{matrix}
				t\\ r_i
			\end{matrix}\right]_{q^{d_i}}
			\le
			c_2^s q^{\frac12 k(n-k)}.
		\end{eqnarray*}
		Since $\sum_{i=1}^s d_ir_i=k$ and $d_i\ge 2$, we have $\sum_{i=1}^s r_i\le  k/2.$ Then the number of $(r_1,\dots,r_s)$ such that $\sum d_i r_i=k$ is at most $(\frac k2+1)^{s}\le (\frac k2+1)^{\ell/2}$. It then follows that
		\begin{eqnarray*}
			\Pr\left[ \cC\tM=\cC\right]
			\le
			\frac{1}{c_1}
			\left(c_2\left(\frac k2+1\right)\right)^{\ell/2}
			q^{-\frac12 k(n-k)}.
		\end{eqnarray*}
		
		We proceed by summing the above probability bound over all nontrivial cosets $\overline{\tM} \in H$ belonging to this subcase, with one chosen representative $\tM=\tD\tP$ for each coset. Roughly, the number of such $\tM$ is at most $n^n(q-1)^n$.
		Define
		\begin{eqnarray*}
			\mathcal{P}_2:=\{ \ell: \ell \ \text{is prime}, \ 3 \le \ell \le n,\ \ell \neq \Char(\gf_q),\ X^{\ell}-\lambda\ \text{has no roots in } \gf_q \}.
		\end{eqnarray*} Then we have
		\begin{eqnarray*}
			S_{\mathrm{noroot}}&:=&\sum_{\ell \in \mathcal{P}_2} 
			\sum_{\substack{1\ne \tM\in H \\ \operatorname{ord}(\tM)=\ell\\ \tP\ne \tI_n}}
			\Pr\left[\cC\tM=\cC\right]\\
			&\le& \sum_{\ell \in \mathcal{P}_2} 
			c_1^{-1} n^n(q-1)^n
			\left(c_2\left(\frac k2+1\right)\right)^{\ell/2}
			q^{-\frac12 k(n-k)} \\
			&\le&
			c_1^{-1}n\cdot n^n(q-1)^n
			\left(c_2\left(\frac k2+1\right)\right)^{n/2}
			q^{-\frac12 k(n-k)}.
		\end{eqnarray*}
		Let $\theta:=\log_q n$, $\eta:=\log_q(q-1)$, and $\rho_n:=\log_q\left(c_2\left(\frac k2+1\right)\right)$. We obtain
		\begin{eqnarray*}
			S_{\mathrm{noroot}}
			\le c_1^{-1} q^{\theta+n\theta+n\eta + \frac{n}{2}\rho_n -\frac{1}{2}k(n-k)}.
		\end{eqnarray*} 
		If $k\le \theta^2$, then we have
		$$\rho_n
		=
		\log_q\left(c_2\left(\frac k2+1\right)\right)
		\le
		\log_q\left(c_2\left(\frac{\theta^2}{2}+1\right)\right)
		= O(\log_q\theta)=
		o(\theta),$$
        and $\theta+n\theta+n\eta+\frac n2\rho_n
		=
		n\theta+o(n\theta)$. Since $k\le n/2$ and $k\ge (2+\varepsilon)\theta$, we have
		$k(n-k)
		\ge
		(2+\varepsilon)\theta\bigl(n-(2+\varepsilon)\theta\bigr)$, and then $\frac12k(n-k)
		\ge
		\left(1+\frac{\varepsilon}{3}\right)n\theta$ for sufficiently large $n$. It follows that
		$$\theta+n\theta+n\eta + \frac{n}{2}\rho_n -\frac{1}{2}k(n-k)
		\le
		n\theta+o(n\theta)
		-
		\left(1+\frac{\varepsilon}{3}\right)n\theta
		\le
		-\frac{\varepsilon}{4}n\theta$$
		for sufficiently large $n$. Therefore, we obtain
		$S_{\mathrm{noroot}}
		\le
		c_1^{-1} q^{-\varepsilon n\theta/4}$, which tends to zero as $n$ tends to infinity.
		If $k > \theta^2$, then we have $k(n-k)/2
		\ge
		n\theta^2/4.$ Since $k\le n/2$, we obtain
		$$\rho_n
		=
		\log_q\left(c_2\left(\frac k2+1\right)\right)
		\le
		\log_q(c_2(n+1))
		=
		\theta+O(1),$$ and hence, $$\theta+n\theta+n\eta+\frac n2\rho_n
		\le
		\theta+n\theta+n\eta+\frac n2(\theta+O(1)) \leq \frac{n\theta^2}{8}$$ for sufficiently large $n$. Therefore, we have 
		$\theta+n\theta+n\eta + \frac{n}{2}\rho_n -\frac{1}{2}k(n-k)
		\le
		-\frac{n\theta^2}{8}$
		for sufficiently large $n$. It follows that
		$S_{\mathrm{noroot}}
		\le
		c_1^{-1} q^{-n\theta^2/8}$, which tends to zero as $n$ tends to infinity.
		
		Combining the estimates for $S_{\mathrm{root}}$ and $S_{\mathrm{noroot}}$, we obtain
		\begin{eqnarray*}
			\sum_{\substack{3\le \ell\le n\\ \ell\ \mathrm{prime}\\ \ell\ne \operatorname{char}(\mathbb F_q)}}
			\sum_{\substack{1\ne \tM\in H\\ \operatorname{ord}(\overline{\tM})=\ell\\ \tP\ne \tI_n}}
			\Pr\bigl[\cC\tM=\cC\bigr] = S_{\mathrm{root}}+S_{\mathrm{noroot}} \to 0,
		\end{eqnarray*}
		as $n\to \infty$.

		\medskip
		\noindent\textbf{Case 2: $\ell=\operatorname{char}(\mathbb F_q) = p $.}
		
		By the Frobenius map $x\mapsto x^{p}$, there exists $\mu\in \mathbb F_q^*$ such that $\mu^{p}=\lambda^{-1}$. Thus, we may use $\mu\tM$ as a representative. We continue to denote this representative by $\tM=\tD\tP$, so that $\tM^{p}=\tI_n$. Let $\tN:=\tM-\tI_n.$ We have $\tN^{p}=(\tM-\tI_n)^p=\tM^p-\tI_n=\bzero$. Suppose that $\tD=\operatorname{diag}(\alpha_1,\dots,\alpha_n)$ with $\alpha_i\in \mathbb F_q^*$, and $\tP$ has $t \geq 1$ $p$-cycles and $f=n-pt$ fixed points, then we write $\tP=\gamma_1\gamma_2\cdots\gamma_t$, where $t\geq 1$ and each $\gamma_j$ is a cycle of length $p$. Let
		\[
		\gamma_j:=(i_{j,1}, i_{j,2}, \cdots, i_{j, p}),\
		W_j:=\langle \mathbf{f}_{j,1}, \mathbf{f}_{j,2}, \cdots, \mathbf{f}_{j,p}\rangle_{\gf_q}.
		\]
		Note that
		\begin{eqnarray*}
			\mathbf{f}_{j,1}\tN^{p-1}
			=
			\mathbf{f}_{j,1}(\tM-\tI_n)^{p-1}
			=
			\sum_{i=0}^{p-1}
			(-1)^{p-1-i}\binom{p-1}{i}\mathbf{f}_{j,i+1}.
		\end{eqnarray*}
		Since $\binom{p-1}{i}\not\equiv 0 \pmod p$, and $\mathbf{f}_{j,1}, \mathbf{f}_{j,2}, \cdots, \mathbf{f}_{j,p}$ are linearly independent over $\gf_q$, we have $\mathbf{f}_{j,1}\tN^{p-1} \neq \bzero$. Moreover, $\tN^p=\bzero$. Since $\dim W_j = p$, the matrix of $\tN|_{W_j}$ under some suitable basis is a nilpotent Jordan block of size $p$. In particular, choose a basis $\by_1=\mathbf{f}_{j,1}, \by_i=\by_1\tN^{i-1}$ for $2\le i\le p,$ the matrix of $\tN|_{W_j}$ under this basis is
		\begin{eqnarray*}
			\begin{pmatrix}
				0&1&0&\cdots&0\\
				0&0&1&\cdots&0\\
				0&0&0&\ddots&0\\
				\vdots&\vdots&\vdots&\ddots&1\\
				0&0&0&\cdots&0
			\end{pmatrix}.
		\end{eqnarray*}
		If $i$ is the fixed point of $\tP$, then $\be_i \tM= \alpha_i \be_i$  for some $\alpha_i \in \gf_q^*$. Since $\tM^p=\tI_n$ and $\Char(\gf_q)=p$, we have $\alpha_i=1$. It follows that $\be_i\tN=\bzero$. Hence, the Jordan type of $\tN$ is $(p^t,1^f)$, where
		$f=n-pt$. Since $\tM=\tN+\tI_n$,  $\cC\tM=\cC$ if and only if $\cC\tN \subseteq \cC$. By Lemma~\ref{lem-Jor}, we have
		\begin{eqnarray*}
			\Pr[ \cC\tM=\cC] \leq \frac{\sum\limits_{\substack{
						\nu_j\ge g_j\ge g_{j+1}\ge 0\\
						\sum g_j=k
				}}
				\prod\limits_{j=1}^p
				q^{g_{j+1}(\nu_j-g_j)}
				\left[\begin{matrix}
					\nu_j-g_{j+1}\\
					g_j-g_{j+1}
				\end{matrix}\right]_q}{\genfrac{[}{]}{0pt}{}{n}{k}_q}.
		\end{eqnarray*}
		By Lemma \ref{lem-Gau}, there exist constants $c_1$, $c_2$ such that
		\begin{eqnarray*}
			\genfrac{[}{]}{0pt}{}{n}{k}_q
			\ge
			c_1 q^{k(n-k)} \mbox{ and } \genfrac{[}{]}{0pt}{}{\nu_j-g_{j+1}}{g_j-g_{j+1}}_q
			\le
			c_2
			q^{(g_j-g_{j+1})(\nu_j-g_j)}.
		\end{eqnarray*}
		Then we have 
		\begin{eqnarray*}
			\prod\limits_{j=1}^p
			q^{g_{j+1}(\nu_j-g_j)}
			\left[\begin{matrix}
				\nu_j-g_{j+1}\\
				g_j-g_{j+1}
			\end{matrix}\right]_q \leq c_2^p q^{\sum_{j=1}^{p}[g_{j+1}(\nu_j-g_j)+(g_j-g_{j+1})(\nu_j-g_j)]} = c_2^p q^{\sum_{j=1}^{p}g_{j}(\nu_j-g_j)}. 
		\end{eqnarray*}
		By setting $\nu_1=t+f$, and
		$\nu_j=t$ for $2\le j\le p$, we have 
		\begin{eqnarray*}
			\sum_{j=1}^p g_j(\nu_j-g_j)
			=
			g_1(t+f-g_1)+\sum_{j=2}^p g_j(t-g_j)=
			t\sum_{j=1}^p g_j
			+
			fg_1
			-
			\sum_{j=1}^p g_j^2 = tk+fg_1-\sum_{j=1}^p g_j^2 
		\end{eqnarray*}
		as $\sum_{j=1}^p g_j=k$. Furthermore, the number of sequences satisfying $g_1\ge\cdots\ge g_p\ge 0$ and
		$\sum_{j=1}^p g_j=k$ is at most $(k+1)^p$. Then we obtain
		\begin{eqnarray*}
			\sum\limits_{\substack{
					\nu_j\ge g_j\ge g_{j+1}\ge 0\\
					\sum g_j=k
			}}
			\prod\limits_{j=1}^p
			q^{g_{j+1}(\nu_j-g_j)}
			\left[\begin{matrix}
				\nu_j-g_{j+1}\\
				g_j-g_{j+1}
			\end{matrix}\right]_q \leq c_2^p(k+1)^p
			\max_{\substack{g_1\ge\cdots\ge g_p\ge 0\\ \sum g_j=k}}
			q^{tk+fg_1-\sum g_j^2}.
		\end{eqnarray*}
		It follows that
		$\Pr[ \cC\tM=\cC] \leq 
		\frac{c_2^p}{c_1}(k+1)^p
		q^{-E_p(k,f)}$,
		where 
		\begin{eqnarray*}
			E_p(k,f)
			=
			\min_{\substack{g_1\ge\cdots\ge g_p\ge 0\\ \sum g_j=k}}
			\left\{
			k(n-k)-tk-fg_1+\sum g_j^2
			\right\}.
		\end{eqnarray*} 
		Let $h:=g_1$. Since $g_1\ge g_2\ge\cdots\ge g_p\ge 0$ and
		$\sum_{j=1}^p g_j=k$, we have $k/p\le h\le k$. By the Cauchy-Schwarz inequality, we obtain
		\begin{eqnarray*}
			\sum_{j=2}^p g_j^2
			\ge \frac{(\sum_{j=2}^p g_j)^2}{p-1} =
			\frac{(k-h)^2}{p-1}.
		\end{eqnarray*}
		Then we have $\sum\limits_{j=1}^p g_j^2
		\ge
		h^2+\frac{(k-h)^2}{p-1}$ and
		\begin{eqnarray*}
			E_p(k,f)
			\ge
			\min_{k/p\le h\le k}
			\left\{
			k(n-k)-tk-fh+h^2+\frac{(k-h)^2}{p-1}
			\right\}.
		\end{eqnarray*}
		Denote by $Q(h):=
		k(n-k)-tk-fh+h^2+\frac{(k-h)^2}{p-1}.$ Then $Q(h)$ is a convex quadratic polynomial in $h$, and $Q'(h)=\frac{2p}{p-1}h-\frac{2k}{p-1}-f$. If $f\ge 2k$, then $Q'(k)=2k-f\le 0$. Since $Q'$ is increasing, we have $Q'(h)\le 0$ for all $h\in[k/p,k]$. It then follows that 
		\begin{eqnarray}\label{eq-fle2k}
			Q(h)\ge Q(k)= k(n-t-f)=k(p-1)t.
		\end{eqnarray}
		If $f \le 2k$, we have
		\begin{eqnarray}\label{eq-fge2k}
			Q(h) \geq Q\left(\frac{2k+(p-1)f}{2p} \right)
			=
			\frac{p-1}{4p}
			\left(4k(n-k)-f^2\right).
		\end{eqnarray}
		Define 
		\begin{eqnarray*}
			\mathcal E_p(k,f):
			=
			\begin{cases}
				k(p-1)t,
				&
				f\ge 2k,\\[6pt]
				\dfrac{p-1}{4p}\left(4k(n-k)-f^2\right),
				&
				f<2k.
			\end{cases}
		\end{eqnarray*}
		It follows from inequalities (\ref{eq-fle2k}) and (\ref{eq-fge2k}) that
		$E_p(k,f)\ge \mathcal E_p(k,f)$, and then
		$$\Pr[ \cC\tM=\cC] \leq 
		\frac{c_2^p}{c_1}(k+1)^p q^{-\mathcal E_p(k,f)}.$$

		We now sum the preceding bound over all nontrivial $\overline{\tM} \in H$ of order $p$ with $\tP \neq \tI_n$, using for each coset the normalized representative $\tM=\tD\tP$ satisfying $\tM^p=\tI_n$. For fixed $t$, the number of permutations $\tP$ with $t$ $p$-cycles and $f=n-pt$ fixed points is at most $n^{pt}$.
		Since $\alpha_{i_{j,1}}\alpha_{i_{j,2}}\cdots \alpha_{i_{j,p}}=1$ for every $p$-cycle $\gamma_j=(i_{j,1}, i_{j,2}, \cdots, i_{j, p})$, there are at most $(q-1)^{(p-1)t}$ choices for $\tD$. Let $\theta:=\log_q n$ and $\eta:=\log_q(q-1)$. Hence, we obtain
		\begin{eqnarray*}
			\sum_{\substack{
					1\neq \overline \tM\in H\\
					\operatorname{ord}(\overline \tM)=p\\
					\tP\neq \tI_n
			}}
			\Pr\left[ \cC\tM=\cC\right] \leq \frac{c_2^p}{c_1}(k+1)^p \sum_{t=1}^{\lfloor n/p\rfloor}q^{t(p\theta+(p-1)\eta)-\mathcal E_p(k,f)}.
		\end{eqnarray*}
		
		For $1 \le t \le \lfloor n/p \rfloor$, the expression $\mathcal E_p(k,f)$ takes different forms depending on whether $t \le \frac{n-2k}{p}$ (equivalently, $f \ge 2k$) or $t > \frac{n-2k}{p}$ (equivalently, $f < 2k$).  We consider these two cases separately. Denote by
		\begin{eqnarray*}
			S_1=\frac{c_2^p}{c_1}(k+1)^p \sum_{\substack{t=1\\ f \ge 2k}}^{\lfloor n/p \rfloor}q^{t(p\theta+(p-1)\eta)-\mathcal E_p(k,f)} \mbox{ and } S_2=\frac{c_2^p}{c_1}(k+1)^p \sum_{\substack{t=1\\ f < 2k}}^{\lfloor n/p \rfloor}q^{t(p\theta+(p-1)\eta)-\mathcal E_p(k,f)}.
		\end{eqnarray*}
		
		If $f\ge 2k$, then $\mathcal E_p(k,f)=k(p-1)t$ and $$t(p\theta+(p-1)\eta)-\mathcal E_p(k,f)=pt\theta+(p-1)t\eta-k(p-1)t
		=
		t\bigl(p\theta+(p-1)\eta-(p-1)k\bigr).$$ 
		Let $B_1:=p\theta+(p-1)\eta-(p-1)k$. 
		Since $k\ge (2+\varepsilon)\theta$, we have $$B_1
		\le
		p\theta+(p-1)\eta-(p-1)(2+\varepsilon)\theta =
		-\bigl(p-2+(p-1)\varepsilon\bigr)\theta+(p-1)\eta.$$
		Since $p \geq 3$ and $\varepsilon > 0$, we have $p-2+(p-1)\varepsilon > 0$. Besides, since $\eta$ is a constant depending only on the fixed $q$, we have
		$(p-1)\eta\le \frac12\bigl(p-2+(p-1)\varepsilon\bigr)\theta$ for sufficiently large $n$. Hence, there exists a constant $c=\frac12\bigl(p-2+(p-1)\varepsilon\bigr)>0$ such that $B_1
		\le -c\theta < 0$ for sufficiently large $n$. 
		Therefore, we have
		\begin{eqnarray*}
			S_1 \leq \frac{c_2^p}{c_1} (k+1)^p \sum_{t \geq 1} q^{tB_1} \leq \frac{c_2^p}{c_1}(k+1)^p \frac{q^{B_1}}{1-q^{B_1}}.
		\end{eqnarray*}
		Consider $p\log_q(k+1)+B_1
		=
		p\log_q(k+1)+p\theta+(p-1)\eta-(p-1)k$. For sufficiently large $n$, the function $p\log_q(k+1)-(p-1)k$ is monotonically decreasing on $[(2+\varepsilon)\theta,n/2]$. Moreover, since $p-2+(p-1)\varepsilon > 0$, we have
		\begin{eqnarray*}
			p\log_q(k+1)+B_1
			&\le&
			p\log_q((2+\varepsilon)\theta+1)
			+p\theta+(p-1)\eta
			-(p-1)(2+\varepsilon)\theta \\
			&=& 
			p\log_q((2+\varepsilon)\theta+1)
			-\bigl(p-2+(p-1)\varepsilon\bigr)\theta
			+(p-1)\eta, 
		\end{eqnarray*} which tends to $-\infty$ as $n$ tends to infinity.
		Therefore, $(k+1)^p q^{B_1}$ tends to zero as $n$ tends to infinity, and immediately, we obtain that $S_1$ tends to zero as $n$ tends to infinity.
		
		If $f<2k$, then $\mathcal E_p(k,f)=\frac{p-1}{4p}\bigl(4k(n-k)-f^2\bigr)$. Since $f=n-pt$, we have $t=\frac{n-f}{p}$ and $\frac{\mathcal E_p(k,f)}{t}
		=
		\frac{p-1}{4}\cdot
		\frac{4k(n-k)-f^2}{n-f}.$ Since $0\le f<2k$ and $k \leq n/2$, we have $nf<2kn\le 4k(n-k).$ It follows that
		\begin{eqnarray*}
			\frac{4k(n-k)-f^2}{n-f}
			-
			\frac{4k(n-k)}{n}
			=
			\frac{f\bigl(4k(n-k)-nf\bigr)}{n(n-f)}
			\ge 0.
		\end{eqnarray*}
		Then we have $\frac{4k(n-k)-f^2}{n-f}
		\ge
		\frac{4k(n-k)}{n}$ and $\mathcal E_p(k,f)
		\ge
		t(p-1) \cdot \frac{k(n-k)}{n}.$  Hence, we obtain
		\begin{eqnarray*}
			t(p\theta+(p-1)\eta)-\mathcal E_p(k,f) \leq t\left(
			p\theta+(p-1)\eta
			-
			(p-1)\frac{k(n-k)}{n}\right).
		\end{eqnarray*} Let $B_2
		:=
		p\theta+(p-1)\eta
		-
		(p-1)\frac{k(n-k)}{n}.$ If $k\le n/4$, then $\frac{k(n-k)}{n}\ge \frac34 k$, and $B_2
		\le
		p\theta+(p-1)\eta-\frac34(p-1)k.$ Since $k\ge (2+\varepsilon)\theta$, we have $$B_2
		\le
		\left[
		p-\frac34(p-1)(2+\varepsilon)
		\right]\theta+(p-1)\eta.$$ Note that $\frac34(p-1)(2+\varepsilon)-p
		=
		\frac{p-3}{2}+\frac34(p-1)\varepsilon>0$ for all $p \geq 3$. Since $\eta$ is a constant depending only on the fixed $q$, we have
		$(p-1)\eta\le \frac12\bigl(\frac34(p-1)(2+\varepsilon)-p\bigr)\theta$ for sufficiently large $n$. Hence, there exists a constant $c'=\frac12\bigl(\frac34(p-1)(2+\varepsilon)-p\bigr)>0$ such that $B_2
		\le -c'\theta < 0$ for sufficiently large $n$. Therefore, we have
		\begin{eqnarray*}
			S_2 \leq \frac{c_2^p}{c_1} (k+1)^p \sum_{t \geq 1} q^{tB_2} \leq \frac{c_2^p}{c_1}(k+1)^p \frac{q^{B_2}}{1-q^{B_2}}.
		\end{eqnarray*}
		Consider $p\log_q(k+1)+B_2 \leq p\log_q(k+1)+p\theta+(p-1)\eta-\frac{3}{4}(p-1)k$. For sufficiently large $n$, the function $p\log_q(k+1)-\frac{3}{4}(p-1)k$ is monotonically decreasing on $[(2+\varepsilon)\theta,n/4]$. It follows that
		\begin{eqnarray*}
			p\log_q(k+1)+B_2
			\le
			p\log_q((2+\varepsilon)\theta+1)
			+
			\left[
			p-\frac34(p-1)(2+\varepsilon)
			\right]\theta
			+(p-1)\eta,
		\end{eqnarray*} which tends to $-\infty$ as $n$ tends to infinity.
		Therefore, $(k+1)^p q^{B_2}$ goes to zero as $n$ goes to infinity, and then $S_2$ goes to zero as $n$ goes to infinity. If $n/2 \ge k > n/4$, then we have $k(n-k)/n > n/8$. Therefore, we obtain
		\begin{eqnarray*}
			p\log_q(k+1)+B_2
			\le
			p\log_q(n+1)+p\theta+(p-1)\eta-\frac{p-1}{8}n,
		\end{eqnarray*} which tends to $-\infty$ as $n$ tends to infinity. Similarly, we have $S_2$ goes to zero as $n$ goes to infinity.
		
		By the discussions above, we have
		\begin{eqnarray*}
			\sum_{\substack{
					1\neq \overline \tM\in H\\
					\operatorname{ord}(\overline \tM)=p\\
					\tP\neq \tI_n
			}}
			\Pr\left[ \cC\tM=\cC\right] \leq S_1+S_2 ,
		\end{eqnarray*} 
        which goes to zero as $n$ goes to infinity.
		
		The desired conclusion follows by combining Cases~1 and~2.
	\end{IEEEproof}
	
	\medskip
	\noindent \textbf{Summary of the proof of the Main Theorem.} When $\tP=\tI_n$, the order $\ell$ of nonidentity element $\overline{\tM}$ must satisfy $\ell \mid (q-1)$, and Theorem~\ref{thm-I} gives us 
	\[\sum_{\substack{\ell \text{ prime,}\\
			\ell~\mid~(q-1)}}
	\sum_{\substack{
			1\neq \overline \tM\in H\\
			\operatorname{ord}(\overline \tM)=\ell\\
			\tP = \tI_n
	}}
	\Pr\left[ \cC\tD=\cC\right] \to 0,
	\]
	as $n\to \infty$. In the case that $\tP \ne \tI_n$ and $\operatorname{ord}(\overline \tM)=2$, by Theorem~\ref{thm-ordertwo} we have \[
	\sum_{\substack{
			1\neq \overline \tM\in H\\
			\operatorname{ord}(\overline \tM)=2\\
			\tP\neq \tI_n}}
	\Pr\left[\cC\tM=\cC\right] \to 0,
	\]as $n\to \infty$. In the case that $\tP \ne \tI_n$ and $\operatorname{ord}(\overline \tM) \geq 3$, by Theorem~\ref{thm-primeorder} we have \[
	\sum_{\substack{\ell \text{ prime}\\ \ell \geq 3}}
	\sum_{\substack{
			1\neq \overline \tM\in H\\
			\operatorname{ord}(\overline \tM) = \ell\\
			\tP\neq \tI_n
	}}
	\Pr\left[ \cC\tM=\cC\right] \to 0,
	\] as $n\to \infty$. Combining all the results above, we now have 
	\begin{eqnarray*}
		\Pr[E_0] & \leq & \sum_{\substack{\ell \text{ prime}\\ \ell ~|~(q-1)}}
		\sum_{\substack{
				1\neq \overline \tM\in H\\
				\operatorname{ord}(\overline \tM)=\ell\\
				\tP = \tI_n
		}}
		\Pr\left[\cC\tD=\cC\right] + \sum_{\substack{
				1\neq \overline \tM\in H\\
				\operatorname{ord}(\overline \tM)=2\\
				\tP\neq \tI_n
		}}
		\Pr\left[\cC\tM=\cC\right] \\
        & & \quad  + \sum_{\substack{\ell \text{ prime}\\ \ell \geq 3}}
		\sum_{\substack{
				1\neq \overline \tM\in H\\
				\operatorname{ord}(\overline \tM) = \ell\\
				\tP\neq \tI_n
		}}
		\Pr\left[ \cC\tM=\cC\right]
	\end{eqnarray*} tends to zero as $n$ goes to infinity.
	This completes the proof of Theorem \ref{main-thm}.
	
	\section{Concluding remarks}\label{sec5}
	
	It is widely assumed in the literature that random $q$-ary linear codes have trivial automorphism groups with high probability. By carefully discussing on all possible cases of invariant subspaces under monomial isometries of prime order with respect to the quotient group, we successfully derived an affirmative result on this assumption upon the condition that $m = \min\{k, n-k\} \ge (2+ \varepsilon) \log_q n$ for any given $\varepsilon > 0$. On the other hand, we proved as well that if $m \le 2 \log_q n + C$, where $C$ is a constant independent of $n$, then the probability that the automorphism group of a random $q$-ary linear code is nontrivial is at least $\frac{1}{2} - \varepsilon$ for large enough $n$ and $\varepsilon > 0$. Both results are of importance not only in their own right, but also in their potential applications in the security analyses of cryptographic schemes based on the LCE-related problems. 
	
	If a random linear code does not admit a trivial automorphism group as assumed in algorithms utilizing the matching codewords framework, the LCE instances may have multiple solutions with non-negligible probability. The codeword-search algorithms may be possibly affected in the following ways: 
	\begin{itemize}
		\item Algorithms requiring multiple matching codeword pairs: different pairs may correspond to distinct equivalence mappings, thereby resulting in incompatible constraints.
		\item Algorithms requiring only a single pair of matching codewords: the isometry obtained through local recovery may fail to extend to a global equivalence one.
	\end{itemize}
	
	It would be naturally interesting to evaluate in a systematic way how exactly the results in this paper affect these related algorithms. In particular, the gap between the two bounds could possibly be narrowed down, which may however require a more refined yet different approach to counting the invariant subspaces. These constitute possible directions for future work.

	\appendix\label{appendix}
	
	\begin{rep}{Lemma}{lem-Jor}
		Let $V:=\gf_q^n$, $p:=\Char(\gf_q)$, and let $\tN: V \to V$ be a nilpotent linear transformation whose Jordan type is $(p^t, 1^f)$. For $1 \leq j \leq p$, define $V_j:=V\tN^{j-1}$, $V_{p+1}=\bzero$, and define the map 
		\begin{eqnarray*}
			\tN_{j} : V_j &\rightarrow& V_{j+1},\\
			\bx &\mapsto& \bx\tN.
		\end{eqnarray*} Let $\cC_j:=\cC\tN^{j-1}$ for $1 \leq j \leq p$, and $\cC_{p+1}=\bzero$. Let $K_j:=\Ker(\tN_{j})$, $\nu_j:=\dim(K_j)$, and $g_j:=\dim(\cC_j \cap K_j)$. 
		Let $T$ denote the number of $k$-dimensional subspaces $\cC \subseteq \gf_q^n$ satisfying $\cC\tN \subseteq \cC$. Then
		\[
		T \leq \sum_{\substack{
				\nu_j\ge g_j\ge g_{j+1}\ge 0\\
				\sum g_j=k
		}}
		\prod_{j=1}^p
		q^{g_{j+1}(\nu_j-g_j)}
		\left[\begin{matrix}
			\nu_j-g_{j+1}\\
			g_j-g_{j+1}
		\end{matrix}\right]_q,
		\]
		where $g_{p+1}=0$ and $\nu_j=\dim(K_j)=\left\{\begin{array}{ll}
			t+f, &  j=1, \\
			t, & 2 \leq j \leq p. 
		\end{array}\right.$
	\end{rep}
	\begin{IEEEproof}
		Since the Jordan type of $\tN$ is $(p^t, 1^f)$, $V$ can be decomposed into $V=W_1\oplus\cdots\oplus W_t\oplus W_F,$ where $\dim(W_i)=p$ and $\tN|_{W_i}$ is a nilpotent Jordan block of size $p$ under some suitable basis, $W_F$ is the $f$-dimensional subspace and $\tN|_{W_F}=\bzero$. In particular, choose a basis $\bg_{i,1},\bg_{i,2},\dots,\bg_{i,p}$ of $W_i$ such that 
		\begin{eqnarray*}
			\bg_{i,1}\tN=\bg_{i,2},\
			\bg_{i,2}\tN=\bg_{i,3},\
			\dots,\
			\bg_{i,p-1}\tN=\bg_{i,p},\
			\bg_{i,p}\tN=\bzero.
		\end{eqnarray*}
		Then we have 
		\begin{eqnarray*}
			V_1=V, \ V_j=V\tN^{j-1}=\bigoplus_{i=1}^t \langle \bg_{i,j}, \cdots, \bg_{i,p} \rangle_{\gf_q} \mbox{ for } 2 \leq j \leq p.
		\end{eqnarray*}
		Moreover, the map $\tN_j$ is surjective by the definition of $V_{j+1}=V_j\tN$. Then the kernel of $\tN_j$ on $W_i\cap V_j$ is spanned by $\bg_{i,p}$ for all $j \geq 1$. In addition, since $W_F\tN=\bzero$, we obtain $W_F \subset V_1$ and it does not appear in $V_j$ for all $j \geq 2$. Hence, we have 
		\begin{eqnarray*}
			\nu_1=\dim K_1=t+f, \ \nu_j=\dim K_j=t \mbox{ for } 2 \leq j \leq p.
		\end{eqnarray*}
		
		Let $\cC \subseteq V$ be a $k$-dimensional subspace satisfying $\cC\tN \subseteq \cC$. Define $\cC_j:=\cC\tN^{j-1}$ for $1 \leq j \leq p$, and $\cC_{p+1}=\bzero$. Then we have $\cC_j\subseteq V_j$ and
		$\cC_{j+1}=\cC_j\tN$. The map $\tN|_{\cC_j}: \cC_j \rightarrow \cC_{j+1}$ is surjective, and its kernel is $\Ker(\tN|_{\cC_j})
		=
		\cC_j\cap K_j.$ Define $g_j:=\dim(\cC_j \cap K_j)$. By the rank-nullity theorem, we obtain $g_j=\dim(\cC_j/\cC_{j+1})=\dim \cC_j-\dim \cC_{j+1}$. Summing over $j$, we obtain
		\begin{eqnarray*}
			\sum_{j=1}^p g_j=\sum_{j=1}^p (\dim \cC_j-\dim \cC_{j+1}) = \dim \cC_1=k.
		\end{eqnarray*}
		Furthermore, for each $1\le j\le p-1$, the map $\tN$ induces a surjective map 
		\begin{eqnarray*}
			\overline{\tN}_j : \cC_j/\cC_{j+1} &\rightarrow& \cC_{j+1}/\cC_{j+2},\\
			\bx + \cC_{j+1} &\mapsto& \bx\tN + \cC_{j+2}.
		\end{eqnarray*}
		Therefore, $\dim(\cC_j/\cC_{j+1})
		\ge
		\dim(\cC_{j+1}/\cC_{j+2}),$ i.e., $g_j\ge g_{j+1}.$ Thus, every $\cC$ satisfying $\cC\tN \subseteq \cC$ determines a sequence 
		\begin{eqnarray*}
			g_1\ge g_2\ge\cdots\ge g_p\ge 0,
			\
			\sum_{j=1}^p g_j=k.
		\end{eqnarray*}
		We now fix such a sequence $(g_1,\dots,g_p)$ and count the possible subspaces $\cC$ by constructing the chain $$\cC_p\subseteq \cC_{p-1}\subseteq\cdots\subseteq \cC_1.$$ 
		Suppose that $\cC_{j+1}\subseteq V_{j+1}$ has already been fixed. We estimate the number of possible $\cC_j \subseteq V_j$ satisfying $$\cC_j\tN=\cC_{j+1},
		\
		\dim(\cC_j\cap K_j)=g_j.$$
		Let $H_j:=\cC_j\cap K_j.$ Since $\cC_{j+1}\subseteq V_{j+1}\subseteq V_j$ and $K_{j+1}=K_j\cap V_{j+1}$, we have $\cC_{j+1}\cap K_j=\cC_{j+1}\cap K_{j+1}.$ Then we have $\dim(\cC_{j+1}\cap K_j)=g_{j+1}.$ Moreover, we obtain $$\cC_{j+1}\cap K_j\subseteq H_j\subseteq K_j,$$ and $\dim H_j=g_j.$ Thus, the number of possible choices for $H_j$ is 
		\begin{eqnarray*}
			\left[\begin{matrix}
				\nu_j-g_{j+1}\\
				g_j-g_{j+1}
			\end{matrix}\right]_q.
		\end{eqnarray*}
		
		After $H_j$ has been fixed, we estimate the number of possible lifts from $\cC_{j+1}$ to $\cC_j$.  Since $\cC_{j+1}\cap K_j\subseteq H_j$ and both are finite-dimensional $\gf_q$-vector spaces, one can choose a complement $H_j^0$ (in $H_j$) such that $H_j=(\cC_{j+1}\cap K_j)\oplus H_j^0$, where $\dim H_j^0=g_j-g_{j+1}.$ 
		Let $\pi_{j+1}: \cC_{j+1}\rightarrow \cC_{j+1}/\cC_{j+2}$ be the natural quotient map. Choose $\bz_1,\dots,\bz_{g_{j+1}}\in \cC_{j+1}$ such that $\pi_{j+1}(\bz_1),\dots,\pi_{j+1}(\bz_{g_{j+1}})$ form a basis of $\cC_{j+1}/\cC_{j+2}$. Equivalently, $\cC_{j+1}=\langle \bz_1,\dots,\bz_{g_{j+1}}\rangle_{\mathbb F_q}+\cC_{j+2}$. For each $\bz_s$, there exists $\bx_s\in \cC_j$ such that $\bx_s\tN=\bz_s$. Now we choose a basis $\bh_1,\dots,\bh_{g_j-g_{j+1}}$ of $H_j^0$, and choose a basis $\by_1, \dots, \by_{\dim \cC_{j+1}}$ of $\cC_{j+1}$. We claim that 
		\begin{eqnarray*}
			\mathcal{B}=\{\by_1, \cdots, \by_{\dim \cC_{j+1}}, \bh_1,\cdots,\bh_{g_j-g_{j+1}}, \bx_1,\cdots,\bx_{g_{j+1}}\}
		\end{eqnarray*}
		is a basis of $\cC_j$. In fact, for any $ \bx \in \cC_j$, we have $\bx \tN \in \cC_{j+1}$, then there exist $a_1,\dots,a_{g_{j+1}}\in \mathbb F_q$ such that $\bx\tN-\sum_{s=1}^{g_{j+1}}a_s \bz_s\in \cC_{j+2}$. Since $\cC_{j+2}=\cC_{j+1}\tN$, there exists $\by \in \cC_{j+1}$ such that $\bx\tN-\sum_{s=1}^{g_{j+1}}a_s \bz_s = \by\tN$. Thus, we have $\left(\bx-\sum_{s=1}^{g_{j+1}}a_s \bx_s -\by\right) \tN = \bzero$, and then $\bx-\sum_{s=1}^{g_{j+1}}a_s \bx_s -\by \in H_j$. Note that $H_j=(\cC_{j+1}\cap K_j)\oplus H_j^0$ and $\cC_{j+1}\cap K_j \subseteq \cC_{j+1}$. Then for any $\bx \in \cC_j$, $\bx$ can be generated by $\mathcal{B}$. Furthermore, we assume that there exist $n_1, n_2, \dots, n_{\dim \cC_{j+1}}, b_1, b_2, \dots, b_{g_j-g_{j+1}}, a_1, a_2, \dots, a_{g_{j+1}} \in \gf_q$ such that 
		\begin{eqnarray*}
			n_1\by_1+\cdots + n_{\dim \cC_{j+1}}\by_{\dim \cC_{j+1}}+b_1\bh_1 + \cdots +b_{g_j-g_{j+1}}\bh_{g_j-g_{j+1}}+ a_1\bx_1 + \cdots +a_{g_{j+1}}\bx_{g_{j+1}}=\bzero.
		\end{eqnarray*} 
		Right-multiplying both sides by $\tN$, we obtain $$\left(n_1\by_1+\cdots + n_{\dim \cC_{j+1}}\by_{\dim \cC_{j+1}}\right)\tN+\sum_{s=1}^{g_{j+1}}a_s \bz_s=\bzero.$$ Then we have $\sum_{s=1}^{g_{j+1}}a_s \bz_s \in \cC_{j+2}$. By the linear independence of $\bz_1+\cC_{j+2},\dots,\bz_{g_{j+1}}+\cC_{j+2}$, we have $a_s=0$ for $s=1, \dots, g_{j+1}$. Hence, we obtain $$n_1\by_1+\cdots + n_{\dim \cC_{j+1}}\by_{\dim \cC_{j+1}}+b_1\bh_1 + \cdots +b_{g_j-g_{j+1}}\bh_{g_j-g_{j+1}}=\bzero.$$ 
		Let $\bc:=n_1\by_1+\cdots + n_{\dim \cC_{j+1}}\by_{\dim \cC_{j+1}} \in \cC_{j+1}$, and $\bh:=b_1\bh_1 + \cdots +b_{g_j-g_{j+1}}\bh_{g_j-g_{j+1}} \in H_j^0$. Since $H_j^0 \subseteq K_j$ and $H_j=(\cC_{j+1}\cap K_j)\oplus H_j^0$, we have $H_j^0\cap \cC_{j+1}=\{\bzero\}$. Then $\bc+\bh=\bzero$ with $\bc \in \cC_{j+1}$ and $\bh \in H_j^0$, yields that $\bc=\bzero$ and $\bh=\bzero$. 
		The linear independence of $\by_1,\dots,\by_{\dim \cC_{j+1}}$ gives $n_i=0$ for $i=1,\dots,\dim \cC_{j+1}$, and the linear independence of $\bh_1,\cdots,\bh_{g_j-g_{j+1}}$ then gives $b_i=0$ for $i=1,\dots,g_j-g_{j+1}$. Therefore, $\mathcal{B}$ is a basis of $\cC_j$, i.e., $\cC_j$ can be generated by $\cC_{j+1}, H_j^0$ together with $\{\bx_s\}_{s=1}^{g_{j+1}}$. 
		
		Thus, the number of possible $\cC_j$ is determined by the number of possible choices of $\{\bx_s\}_{s=1}^{g_{j+1}}$ after fixing $\cC_{j+1}$ and $H_j$. Suppose $\bx_s^0$ is one solution of $\bx\tN=\bz_s$. Then the solution set is $\bx_s^0+K_j$. In addition, if two solutions $\bx_s$ and $\bx_s'$ satisfy $\bx_s'=\bx_s+\bh$ for some $\bh \in H_j$, then we have
		\begin{eqnarray*}
			\langle \cC_{j+1},H_j,\bx_1,\dots,\bx_s,\dots,\bx_{g_{j+1}}\rangle
			=
			\langle \cC_{j+1},H_j,\bx_1,\dots,\bx_s',\dots,\bx_{g_{j+1}}\rangle,
		\end{eqnarray*}
		i.e., they generate the same subspace $\cC_j$. Hence, the number of effective choices for each $\bx_s$ is at most
		\[
		q^{\dim(K_j/H_j)}=q^{\nu_j-g_j},
		\]
		and the number of effective choices for $\{\bx_s\}_{s=1}^{g_{j+1}}$ is at most $q^{g_{j+1}(\nu_j-g_j)}$.
		Combining this estimate with the number of choices for $H_j$, we obtain that the number of possible $\cC_j$ is at most 
		\begin{eqnarray*}
			q^{g_{j+1}(\nu_j-g_j)}
			\left[\begin{matrix}
				\nu_j-g_{j+1}\\
				g_j-g_{j+1}
			\end{matrix}\right]_q
		\end{eqnarray*} based on fixed $\cC_{j+1}$.
		
		By the above discussions, we have 
		\begin{eqnarray}\label{eq-Jor}
			T \leq \sum_{\substack{
					\nu_j\ge g_j\ge g_{j+1}\ge 0\\
					\sum g_j=k
			}}
			\prod_{j=1}^p
			q^{g_{j+1}(\nu_j-g_j)}
			\left[\begin{matrix}
				\nu_j-g_{j+1}\\
				g_j-g_{j+1}
			\end{matrix}\right]_q,
		\end{eqnarray}
		where $g_{p+1}=0$ and $\nu_j=\dim(K_j)=\left\{\begin{array}{ll}
			t+f, &  j=1, \\
			t, & 2 \leq j \leq p. 
		\end{array}\right.$
	\end{IEEEproof}
\begin{remark}
Because our proof of the main theorem needs to address the characteristic case, namely, $p=\Char(\gf_q)$, we impose this condition in Lemma \ref{lem-Jor}. In fact, the conclusions of Lemma \ref{lem-Jor} still hold for every integer $p \ge 2$.
\end{remark}
	
	\begin{rep}{Corollary}{cor-ordertwo}
		Let $\tN: \gf_q^n \to \gf_q^n$ be a nonzero linear transformation satisfying
		\[
		\tN^2=\bzero,\
		\dim (\operatorname{Im}(\tN))=t, \textrm{ and }
		\dim(\Ker (\tN))=n-t.
		\]
		Let $T$ denote the number of $k$-dimensional subspaces $\cC \subseteq \gf_q^n$ satisfying $\cC\tN \subseteq \cC$. Then
		\[
		T \leq \sum_{r=0}^{\min\{t,\lfloor k/2\rfloor\}}
		\genfrac{[}{]}{0pt}{}{t}{r}_q
		\genfrac{[}{]}{0pt}{}{n-t-r}{k-2r}_q
		q^{(n-t-k+r)r}.
		\] 
	\end{rep}
	\begin{IEEEproof}
		Since $\tN$ is a nonzero linear transformation satisfying
		$
		\tN^2=\bzero$, 
		$\dim (\operatorname{Im}(\tN))=t$, and
		$\dim(\Ker (\tN))=n-t
		$, it follows that the Jordan type of $\tN$ is $(2^t, 1^{n-2t})$.
		By letting $p=2$ in Lemma \ref{lem-Jor}, we have $\nu_1=n-t$ and $\nu_2=t$. Moreover, assume that the dimension of $\cC \cap \Ker(\tN)$ is $k-r$, i.e., $k-r=g_1=\dim(\cC \cap \Ker(\tN))$, then $g_2=r$ and $g_3=0$. It then follows from inequality (\ref{eq-Jor}) that 
		\[
		T \leq \sum_{r=0}^{\min\{t,\lfloor k/2\rfloor\}}
		\genfrac{[}{]}{0pt}{}{t}{r}_q
		\genfrac{[}{]}{0pt}{}{n-t-r}{k-2r}_q
		q^{(n-t-k+r)r}.
		\] 
	\end{IEEEproof}
\end{document}